\documentclass[aapnotag,preprint]{imsart}

\usepackage{amsfonts,psfrag,epsfig}
\usepackage{color}
\usepackage{amsmath,amssymb}
\usepackage[amsmath,thmmarks]{ntheorem}
\newtheorem{theorem}{Theorem}
\newtheorem{lemma}[theorem]{Lemma}

\newtheorem{proposition}[theorem]{Proposition}

\newtheorem{definition}{Definition}

{
\theoremstyle{nonumberplain}
\theoremsymbol{\ensuremath{\Box}}
\theoremheaderfont{\normalfont\itshape}
\theorembodyfont{\normalfont}
\newtheorem{proof}{Proof.}

\theoremstyle{empty}

}

\newcommand{\norm}[1]{\left\lVert #1 \right\rVert}
\newcommand{\abs}[1]{\left\lvert #1 \right\rvert}

\newcommand{\bind}{\bold{I}}

\newcommand{\p}[1]{\left( #1 \right)}

\newcommand{\bW}{\boldsymbol{W}}

\newcommand{\sig}{\boldsymbol{\sigma}}
\newcommand{\bsigma}{\boldsymbol{\sigma}}
\newcommand{\bdelta}{\boldsymbol{\delta}}
\newcommand{\bdel}{\bdelta}%{\boldsymbol{\delta}}
\newcommand{\rhO}{\boldsymbol{\rho}}
\newcommand{\lamb}{\boldsymbol{\lambda}}
\newcommand{\Lamb}{\boldsymbol{\Lambda}}

\newcommand{\MW}{\textsf{MW}}

\newcommand{\cN}{\mathcal{N}}

\newcommand{\bphi}{\boldsymbol{\phi}}

\newcommand{\beq}{\begin{eqnarray}}
\newcommand{\eeq}{\end{eqnarray}}
\newcommand{\beqn}{\begin{equation}}
\newcommand{\eeqn}{\end{equation}}
\newcommand{\R}{\mathbb{R}}
\newcommand{\N}{\mathbb{N}}
\newcommand{\Z}{\mathbb{Z}}

\newcommand{\Rp}{\mathbb{R}_+}
\newcommand{\Zp}{\mathbb{Z}_+}

\newcommand{\beps}{\varepsilon}

\newcommand{\bzero}{\mathbf{0}}
\newcommand{\bone}{\mathbf{1}}

\newcommand{\mc}{\mathcal}
\newcommand{\mb}{\mathbf}

\newcommand{\LN}{\textbf{LN}}

\renewcommand{\hat}{\widehat}
\newcommand{\cM}{\mc{M}}
\newcommand{\sX}{{\sf X}}
\newcommand{\cX}{\mc{X}}

\newcommand{\cR}{\mc{R}}

\newcommand{\cB}{\mc{B}}

\newcommand{\cI}{\mc{I}}

\newcommand{\bu}{\mb{u}}
\newcommand{\bv}{\mb{v}}

\newcommand{\bx}{\mb{x}}

\newcommand{\bz}{\mb{z}}

\newcommand{\bQ}{\mathbf{Q}}
\newcommand{\bR}{\mathbf{R}}

\newcommand{\E}{\mathbb{E}}

\newcommand{\ind}[1]{\boldsymbol{1}_{[#1]}}  %Indicator
\newcommand{\x}{\boldsymbol{x}}
\newcommand{\y}{\boldsymbol{y}}
\newcommand{\z}{\boldsymbol{z}}
\newcommand{\X}{{\cal X}}              %\chi is a lowercase letter!
\DeclareMathOperator{\Conv}{\textsf{Conv}}

\begin{document}
\begin{frontmatter}

\title{Randomized  Scheduling Algorithm  for Queueing Networks}

\runtitle{Randomized Network Scheduling}

\begin{aug}
  \author{\fnms{D.} \snm{Shah} \quad \fnms{J.} \snm{Shin}\thanksref{shin}\ead[label=e1]{jinwoos@mit.edu}}
  \runauthor{Shah \& Shin}
  \affiliation{Massachusetts Institute of Technology}
\thankstext{shin}{Both authors are with the Laboratory for
Information and Decision Systems at MIT. DS and JS are with
the department of EECS and Mathematics, respectively.
Authors' email addresses: {\tt \{devavrat, jinwoos\}@mit.edu}}
\end{aug}

\vspace{.1in}

\begin{abstract}

There has recently been considerable interests in design of
low-complexity, myopic, distributed and stable scheduling
policies for constrained queueing network models that arise
in the context of emerging communication networks. Here, we
consider two representative models. One, a model for
the collection of wireless nodes communicating through a
shared medium, that represents randomly varying number of
packets in the queues at the nodes of networks. Two, a
buffered circuit switched network model for an optical
core of future Internet, to capture the randomness in
calls or flows present in the network. The maximum weight
scheduling policy proposed by Tassiulas and Ephremide
\cite{TE92} leads to a myopic and stable policy for
the packet-level wireless network model. But computationally
it is very expensive (NP-hard) and centralized. It is
not applicable to the buffered circuit switched network due
to the requirement of non-premption of the calls in the
service. As the main contribution of this paper, we
present a stable scheduling algorithm for both of these
models. The algorithm is myopic, distributed and performs
few logical operations at each node per unit time.

%Our
%arguments for establishing stability (positive Harris
%recurrence) is based on identifying an appropriate
%Lyapunov

\end{abstract}

\begin{keyword}[class=AMS]
\kwd[Primary ]{60K20}
\kwd{68M12}
\kwd[; Secondary ]{68M20}
\end{keyword}

\begin{keyword}
\kwd{Wireless Medium Access}
\kwd{Buffered Circuit Switched Network}
\kwd{Aloha}
\kwd{Stability}
\kwd{Scheduling}
\kwd{Mixing time}
\kwd{Slowly Varying Markov Chain}
\end{keyword}

\end{frontmatter}

\section{Introduction}

The primary task of a communication network architect is
to provision as well as utilize network resources efficiently
to satisfy the demands imposed on it. The main algorithmic
problem is that of allocating or scheduling resources among
various entities or data units, e.g. packets, flows, that
are contending to access them. In recent years, the question
of designing a simple, myopic, distributed and high-performance
(aka stable) scheduling algorithm has received considerable
interest in the context of emerging communication network
models. Two such models that we consider this paper are that
of a wireless network and a buffered circuit switched network.

The wireless network consists of wireless transmission capable nodes.
Each node receives exogenous demand in form of packets. These nodes
communicate these packets through a shared wireless medium. Hence
their simultaneous transmission may contend with each
other. The purpose of a scheduling algorithm is to resolve
these contentions among transmitting nodes so as to utilize
the wireless network bandwidth efficiently while keeping the
queues at nodes finite. Naturally the desired scheduling algorithm
should be distributed, simple/low-complexity and myopic (i.e. utilize only
the network state information like queue-sizes).

The buffered circuit switched network can be utilized to model
the dynamics of flows or calls in an optical core of future
Internet. Here a link capacitated network is given with a collection
of end-to-end routes. At the ingress (i.e. input or entry point)
of each route, calls arriving as per exogenous
process are buffered or queued. Each such call
desires resources on each link of its route for a random amount
of time duration. Due to link capacity constraints, calls of
routes sharing links contend for resources. And, a scheduling algorithm
is required to resolve this contention so as to utilize the network
links efficiently while keeping buffers or queues at ingress of routes
finite. Again, the scheduling algorithm is desired to be
distributed, simple and myopic.

An important scheduling algorithm is the maximum weight
policy that was proposed by Tassiulas and Ephremides \cite{TE92}.
It was proposed in the context of a packet queueing network model
with generic scheduling constraints. It is primarily applicable
in a scenario where scheduling decisions are synchronized or
made every discrete time. It suggests scheduling queues, subject
to constraints, that have the maximum net weight at each time
step with the weight of a queue being its queue-size. They
established throughput optimality property of this algorithm
for this general class of networks. Further, this algorithm, as the
description suggests, is myopic. Due to the general applicability
and myopic nature, this algorithm and its variants have received
a lot of attention in recent years, e.g.
\cite{MAW,DB,stolyar,SW06,dailin,SW07}.

The maximum weight algorithm provides a myopic and stable scheduling
algorithm for the wireless network model. However, it requires
solving a combinatorial optimization problem, the maximum
weight independent set problem, to come up with a schedule every
time. And the problem of finding a maximum weight independent set is
known to be NP-hard as well as hard to approximate in general
\cite{IS}. To address this concern, there has been
a long line of research conducted to devise implementable approximations
of the maximum weight scheduling algorithm, e.g. \cite{islip,tassiulas98,G-P-S,DW04,MSZ06}.
A comprehensive survey of such maximum weight inspired and other
algorithmic approaches that have been studied over more than four decades
in the context of wireless networks can be found in \cite{RSS09,LSSW}.

In the context of buffered circuit switched network, calls have
random service requirement. Therefore, scheduling decisions can not be
synchronized. Therefore, the maximum weight scheduling algorithm is not
applicable. To the best of our knowledge, no other myopic and stable
algorithm is known for this network model.

\subsection{Contributions}

We propose a scheduling algorithm for both wireless and buffered
circuit switched network model. The algorithm utilizes only local,
queue-size information to make scheduling decisions. That is, the
algorithm is myopic and distributed. It requires each queue
(or node) in the network to perform few (literally, constant) logical
operations per scheduling decision. We establish that it is throughput
optimal. That is, the network Markov process is positive Harris
recurrent as long as the network is under-loaded (or not overloaded).

Philosophically, our algorithm design is motivated by a certain product-form
distribution that can be characterized as the stationary distribution of
a simple and distributed Markovian dynamics over the space of schedules. For
the wireless network, it corresponds to the known Glauber dynamics
(cf. \cite{Glauber}) over the space of independent sets of the wireless
network interference graph; for the buffered circuit switched network,
it corresponds to the known stochastic loss network (cf. \cite{Kelly}).

To establish the stability property of the algorithm, we exhibit an
appropriate Lyapunov function. This, along with standard machinery
based on identifying an appropriate `petit set', leads to the positive Harris
recurrence property of the network Markov process.  Technically, this is the most
challenging part of our result. It requires proving an effective
`time scale separation' between the network queuing dynamics and
the scheduling dynamics induced by the algorithm. To make this
possible, we use an appropriately slowly increasing
function ($\log\log (\cdot + e)$) of queue-size as weight
in the scheduling algorithm. Subsequently, the time scale separation
follows by studying the mixing property of a specific time varying
Markov chain over the space of schedules.

We note that use of Lyapunov function for establishing stability is somewhat classical now (for example, see \cite{TE92, stolyar, SW06}). Usually difficulty lies in finding an appropriate candidate function followed by establishing that it is indeed a ``Lyapunov'' function.

\subsection{Organization}

We start by describing two network models, the wireless network and
the buffered circuit switched network in Section \ref{sec:model}. We
formally introduce the problem of scheduling and performance metric
for scheduling algorithms. The maximum weight scheduling algorithm
is described as well. Our randomized algorithm and its throughput
optimality for both network models are presented in Section \ref{sec:main}.
The paper beyond Section \ref{sec:main} is dedicated to establishing
the throughput optimality. Necessary technical preliminaries
are presented in Section \ref{sec:prelim}. Here we relate our algorithm
for both models with appropriate reversible Markov chains on the space
of schedules and state useful properties of these Markov chains.
We also describe known facts about the positive Harris recurrence
as well as state the known Lyapunov drift criteria, to establish
positive Harris recurrence. Detailed proofs of our main results are
presented in Section \ref{sec:mainproof}.

\section{Setup}\label{sec:model}

\subsection{Wireless Network}

We consider a {\em single-hop} wireless network of $n$ queues.
Queues receive work as per exogenous arrivals and work leaves
the system upon receiving service. Specifically, let
$Q_i(t) \in \R_+ = \{ x \in \R: x\geq 0\}$
denote the amount of work in the $i$th queue at time
$t\in \R_+$ and $\bQ(t)=[Q_i(t)]_{1{\le}i{\le}n}$;
initially $t=0$ and $\bQ(0) = \bzero$\footnote{Bold letters are reserved
for vectors; $\bzero, \bone$ represent vectors of all $0$s \& all $1$s
respectively.}.
Work arrives to each queue in terms of unit-sized packets
as per a discrete-time process. Let $A_i(s,t)$ denote the
amount of work arriving to queue $i$ in time interval $[s,t]$
for $0\leq s < t$. For simplicity, assume that
for each $i$, $A_i(\cdot)$ is an independent Bernoulli process
with parameter $\lambda_i$, where $A_i(\tau) \stackrel{\triangle}{=} A_i(0,\tau)$.
That is, $A_i(\tau+1)-A_i(\tau) \in \{0,1\}$ and
$\Pr(A_i(\tau+1)-A_i(\tau) = 1) = \lambda_i$ for all
$i$ and $\tau\in \Zp = \{ k \in \Z : k \geq 0\}$.  Denote the arrival rate vector as $\lamb =
[\lambda_i]_{1\le i\le n}$. We assume that arrivals happen at
the end of a time slot.

The work from queues is served at the unit rate, but subject to
{\em interference} constraints. Specifically, let $G = (V,E)$
denote the inference graph between the $n$ queues, represented by
vertices $V = \{1,\dots n\}$ and edges $E$: an $(i,j) \in E$ implies
that queues $i$ and $j$ can not transmit simultaneously since
their transmission {\em interfere} with each other. Formally,
let $\sigma_i(t) \in \{0,1\}$ denotes whether the queue $i$
is transmitting at time $t$, i.e. work in queue $i$ is being served
at unit rate at time $t$ and $\sig(t) = [\sigma_i(t)]$. Then, it must
be that for $t \in \R_+$,
$$\sig(t) \in \cI(G) \stackrel{\Delta}{=}
\{ \rhO = [\rho_i] \in \{0,1\}^n :
   \rho_i + \rho_j \le 1\text{ for all }(i,j) \in E \}.$$
The total amount of work served at queue $i$ in time interval $[s,t]$ is
$$D_i(s,t) = \int_{s}^t \sigma_i(y) \bind_{\{Q_i(y) > 0\}} dy,$$
where $\bind_{\{x\}}$ denotes the indicator function.

In summary, the above model induces the following queueing dynamics:
for any $0 \leq s < t$ and $1\leq i\leq n$,
$$ Q_i(t) = Q_i(s) - \int_{s}^t \sigma_i(y) \bind_{\{Q_i(y) > 0\}} dy + A_i(s,t). $$

%To serve packets in queues, each queue can request the network for the availability of transmission i.e.
%``listen'' to the medium whether its interfering neighbors are transmitting or not.
%Depending on the response of the network, the queue $i$ can start transmitting (or served) i.e. become active.
%The total amount of packets (or work) served at the queue $i$ is decided by the total amount of time when
%$i$ is active.

\subsection{Buffered Circuit Switched Network}

We consider a buffered circuit switched network. Here
the network is represented by  a capacitated graph
$G=(V,E)$ with $V$ being vertices, $E \subset V \times V$
being links (or edges) with each link $e \in E$ having
a finite integral capacity $C_e \in \N$. This network
is accessed by a fixed set of $n$ routes $R_1, \dots, R_n$;
each route is a collection of interconnected links. At
each route $R_i$, flows arrive as per an exogenous arrival
process. For simplicity, we assume it to be an independent
Possion process of rate $\lambda_i$ and let $A_i(s,t)$ denote
total number of flow arrivals to route $R_i$ in time interval
$[s,t]$. Upon arrival of a flow to route $R_i$, it joins the queue or
buffer at the ingress of $R_i$. Let $Q_i(t)$ denote the number
of flows in this queue at time $t$; initially $t=0$ and $Q_i(0)=0$.

Each flow arriving to $R_i$, comes with the service requirement of
using unit capacity simultaneously on all the links of $R_i$
for a time duration -- it is assumed to be distributed independently
as per Exponential of unit mean. Now a flow in the queue of
route $R_i$ can get simultaneous possession of links along
route $R_i$ in the network at time $t$, if there is a unit
capacity available at {\em all} of these links. To this end,
let $z_i(t)$ denote the number of flows that are {\em active}
along route $R_i$, i.e. posses links along the route $R_i$.
Then, by capacity constraints on the links of the network,
it must be that $z(t) = [z_i(t)]$ satisfies
$$ \z(t)\in \X \stackrel{\Delta}{=}\{\z=[z_i]\in \Zp^n : \sum_{i:e\in R_i} z_i \leq C_e, ~~\forall~ e\in E\}.$$
This represents the scheduling constraints of the circuit switched
network model similar to the interference constraints of the wireless
network model.

Finally, a flow active on route $R_i$, departs the network
after the completion of its service requirement and frees unit
capacity on the links of $R_i$. Let $D_i(s,t)$ denote the
number of flows which are served (hence leave the system)
in time interval $[s,t]$.

\subsection{Scheduling Algorithm \& Performance Metric}

In both models described above, the scheduling is the
key operational question. In the wireless network, queues
need to decide which of them transmit subject to interference
constraints. In the circuit switched network, queues need to
agree on which flows becomes active subject to network
capacity constraints. And, a {\em scheduling algorithm}
is required to make these decisions every time.

In wireless network, the scheduling algorithm decides
the schedule $\bsigma(t) \in \cI(G)$ at each time $t$.
We are interested in {\em distributed} scheduling algorithms,
i.e. queue $i$ decides $\sigma_i(t)$ using its local information,
such as its queue-size $Q_i(t)$. We assume that queues
have instantaneous {\em carrier sensing} information, i.e.
if a queue (or node) $j$ starts transmitting at time $t$,
then all neighboring queues can {\em listen} to this
transmission immediately.

In buffered circuit switched network, the scheduling
algorithm decides active flows or schedules $\bz(t)$ at
time $t$. Again, our interest is in {\em distributed}
scheduling algorithms, i.e. queue at ingress of route $R_i$
decides $z_i(t)$ using its local information. Each
queue (or route) can obtain instantaneous information
on whether all links along its route have unit capacity
available or not.

In summary, both models need scheduling algorithms
to decide when each queue (or its ingress port) will request
the network for availability of resources; upon a positive
answer (or successful request) from the network, the queue acquires
network resources for certain amount of time. And,
these algorithm need to be based on local information.

From the perspective of network performance, we would like the
scheduling algorithm to be such that the queues in network remain
as small as possible for the largest possible range of arrival
rate vectors.  To formalize this notion of performance, we
define the capacity regions for both of these models. Let $\Lamb_{w}$
be the capacity region of the wireless network model defined as
\begin{eqnarray}
 \Lamb_{w} & = & \Conv(\cI(G))\nonumber\\
 &=&\left\{ \y \in \Rp^n : \y \leq \sum_{\sig \in \cI(G)} \alpha_{\sig} \sig,
~\mbox{with}~\alpha_{\sig} \geq 0,~\mbox{and}~\sum_{\sig \in \cI(G)} \alpha_{\sig} \leq 1 \right\}.
\label{eq:c1}
\end{eqnarray}
And let $\Lamb_{cs}$ be the capacity region of the buffered circuit switched
network defined as
\begin{eqnarray}
 \Lamb_{cs} & = &  \Conv(\X)\nonumber\\
 &=&\left\{ \y \in \Rp^n : \y \leq \sum_{\z \in \X} \alpha_{\z} \z,
~\mbox{with}~\alpha_{\z} \geq 0,~\mbox{and}~\sum_{\z \in \X} \alpha_{\z} \leq 1 \right\}.
\label{eq:c2}
\end{eqnarray}
Intuitively, these bounds of capacity regions comes from
the fact that any algorithm produces the `service rate' from $\cI(G)$ (or $\X$) each
time and hence the time average of the service rate induced by
any algorithm must belong to its convex hull. Therefore, if arrival rates
$\lamb$ can be `served well' by any algorithm then it must belong to $\Conv(\cI(G))$ (or $\Conv(\X)$).

Motivated by this, we call an arrival rate vector $\lamb$ admissible
if $\lamb \in \Lamb$, and say that an arrival rate vector $\lamb$ is strictly admissible if
$\lamb \in {\Lamb}^o$, where $\Lamb^o$ is the interior of $\Lamb$ formally
defined as
$$ \Lamb^o = \left\{ \lamb \in \mathbb{R}^n_{+} :
\lamb < \lamb^* \text{ componentwise, for some }\lamb^* \in \Lamb \right\}. $$
%\[{\Lamb}^o = \left\{ \lamb \in \Lamb : \lamb \le \sum_{\sig \in \FS}{\alpha_{\sig} \sig}\text{, with }\sum{\alpha_{\sig}} < 1\text{ and }\alpha_{\sig} \ge 0\text{ for all }\sig \right\}\]
Equivalently, we may say that the network is \emph{under-loaded}.
Now we are ready to define a performance metric for a scheduling
algorithm. Specifically, we desire the scheduling algorithm to be
throughput optimal as defined below.

\begin{definition}[throughput optimal] {\em A scheduling algorithm is
called \\ { throughput optimal}, or {stable}, or
providing {100\% throughput}, if for  any $\lamb \in \Lamb^o$ the
(appropriately defined) underlying network Markov process is
\emph{positive (Harris) recurrent}}.
\end{definition}

%\begin{definition}[Rate-stable] We call a sch\-eduling algorithm
%\emph{rate-stable},
%if for any $\lamb \in \Lamb^o$ the empirical departure rate converges to the arrival rate
%with probability 1 i.e.
%$$\forall i,~~\lim_{t\rightarrow \infty} \frac1{t}D_i(t)=\lambda_i~~\mbox{with probability}~1.$$
%\end{definition}

%Note that the rate stability is a weaker notion than the throughput optimality
%since the latter one implies that the queue-sizes remain finite as the time goes to the infinity while
%the former one may not guarantee it.

\subsection{The {\MW} Algorithm}

Here we describe a popular algorithm known as the maximum weight or
in short \MW~algorithm that was proposed by Tassiulas and Ephremides \cite{TE92}.
It is throughput optimal for a large class of network models. The
algorithm readily applies to the wireless network model. However,
it does not apply (exactly or any variant of it) in the case of
circuit switched network. Further, this algorithm requires solving
a hard combinatorial problem each time slot, e.g. maximum weight independent
set for wireless network, which is NP-hard in general. Therefore, it's far
from being practically useful. In a nutshell, the randomized
algorithm proposed in this paper will overcome these drawbacks
of the \MW~algorithm while retaining the throughput optimality
property. For completeness, next we provide a brief description
of the \MW~algorithm.

In the wireless network model, the \MW~algorithm chooses a
schedule $\sig(\tau) \in \cI(G)$ every time step $\tau \in \Zp$
as follows\footnote{Here and everywhere else, we use notation $\bu \cdot \bv = \sum_{i=1}^d u_i v_i$
for any $d$-dimensional vectors $\bu, \bv \in \R^d$. That is,
$\bQ(\tau)\cdot \rhO =\sum_i Q_i(\tau)\cdot\rho_i$.}:
$$ \sig(\tau) \in \arg\max_{\rhO \in \cI(G)} \bQ(\tau)\cdot \rhO. $$
In other words, the algorithm changes its decision once in unit
time utilizing the information $\bQ(\tau)$. The maximum weight
property allows one to establish positive recurrence by means of
Lyapunov drift criteria (see Lemma \ref{lem:two}) when the arrival
rate is admissible, i.e. $\lamb \in \Lamb^o_{w}$. However, as indicated
above picking such a schedule every time is computationally burdensome.
A natural generalization of this, called \MW-$f$ algorithm, that uses
weight $f(Q_i(\cdot))$ instead of $Q_i(\cdot)$ for an increasing non-negative
function $f$ also leads to throughput optimality (cf. see \cite{stolyar, SW06, SW07}).

For the buffered circuit switched network model, the \MW~algorithm is
not applicable. To understand this, consider the following. The
\MW~algorithm would require the network to schedule active flows as
$\z(\tau)\in\X$ where
$$ \z(\tau) \in \arg\max_{\z \in \X} \bQ(\tau)\cdot\z.$$
This will require the algorithm to possibly preempt some of
active flows without the completion of their service requirement. And
this is not allowed in this model.

\section{Main Result: Simple \& Efficient Randomized Algorithms}\label{sec:main}

As stated above, the \MW~algorithm is not practical for wireless network
and is not applicable to circuit switched network. However, it has the
desirable throughput optimality property. As the main result of this
paper, we provide a simple, randomized algorithm that is applicable to
both wireless and circuit switched network as well as it's throughput
optimal. The algorithm requires each node (or queue) to perform only
a few logical operations at each time step, it's distributed and effectively
it `simulates' the \MW-$f$ algorithm for an appropriate choice of $f$.
In that sense, it's a simple, randomized, distributed implementation
of the \MW~algorithm.

In what follows, we shall describe algorithms for wireless network
and buffered circuit switched network respectively. We will state
their throughput optimality property. While these algorithms seem
different, philosophically they are very similar -- also, witnessed
in the commonality in their proofs.

\subsection{Algorithm for Wireless Network}\label{ssec:algo1}

Let $t \in \R_+$ denote the time index and $\bW(t) = [W_i(t)] \in \R_+^n$
be the vector of weights at the $n$ queues. The $\bW(t)$ will be
a function of $\bQ(t)$ to be determined later. In a nutshell, the
algorithm described below will choose a schedule $\sig(t) \in \cI(G)$
so that the weight, $\bW(t)\cdot \sig(t)$, is as large as possible.

The algorithm is randomized and asynchronous.  Each node (or queue)
has an independent Exponential clock of rate $1$ (i.e. Poisson process
of rate $1$). Let the $k$th tick of the clock of node $i$ happen at time
$T_k^i$; $T^i_0 = 0$ for all $i$. By definition $T_{k+1}^i - T_k^i, k \geq 0,$
are i.i.d. mean $1$ Exponential random variables. Each node changes
its scheduling decision only at its clock ticks. That is, for
node $i$ the $\sigma_i(t)$ remains constant for $t \in (T^i_k, T^i_{k+1}]$.
Clearly, with probability $1$ no two clock ticks across nodes happen at the same time.

Initially, we assume that $\sigma_i(0) = 0$ for all $i$.
The node $i$ at the $k$th clock tick, $t = T^i_k$, {\em listens} to
the medium and does the following:

\begin{itemize}

\item[$\circ$] If any neighbor of $i$ is transmitting, i.e.
$\sigma_j(t) = 1$ for some $j\in \cN(i)=\{j':(i,j')\in E\}$,
then set $\sigma_i(t^+) = 0$.

\item[] ~~

\item[$\circ$] Else, set
$$ \sigma_i(t^+) = \begin{cases} 1 & \quad \text{with probability} \quad
\frac{\exp(W_i(t))}{1+\exp(W_i(t))} \\
 0 & \quad\text{otherwise.}
 \end{cases}
 $$
\end{itemize}

Here, we assume that if $\sigma_i(t) = 1$, then node $i$ will
always transmit data irrespective of the value of $Q_i(t)$
so that the neighbors of node $i$ can infer $\sigma_i(t)$
by {\em listening} to the medium.

\subsubsection{Throughout Optimality}

The above described algorithm for wireless network is throughput
optimal for an appropriate choice of weight $\bW(t)$.
Define weight $W_i(t)$ at node $i$ in the algorithm
for wireless network as
\begin{equation}
W_i(t) = \max\left\{f(Q_i(\lfloor t\rfloor)),\sqrt{f(Q_{\max}(\lfloor t\rfloor))}\right\},
\label{eq:weight1}
\end{equation}
where\footnote{Unless stated otherwise, here and everywhere else
the $\log(\cdot)$ is natural logarithm, i.e. base $e$.} $f(x) = \log \log (x+e)$
and $Q_{\max}(\cdot) = \max_{i} Q_i(\cdot)$. The non-local information
of $Q_{\max}(\lfloor t\rfloor))$ can be replaced by its approximate
estimation that can computed through a very simple distributed algorithm.
This does not alter the throughput optimality property of the algorithm.
A discussion is provided in Section \ref{sec:discuss}.  We state the
following property of the algorithm.
\begin{theorem}\label{thm:main1}
Suppose the algorithm of Section \ref{ssec:algo1}
uses the weight as per \eqref{eq:weight1}. Then, for any $\lamb \in \Lamb_w^o$,
the network Markov process is positive Harris recurrent.
\end{theorem}
In this paper, Theorem \ref{thm:main1} (as well as Theorem \ref{thm:main2})
is established for the choice of $f(x) = \log\log (x + e)$. However, the proof technique of this paper extends naturally for any choice of $f : \Rp \to \Rp$ that satisfies the following conditions: $f(0) = 0$, $f$ is a monotonically strictly increasing function, $\lim_{x\to\infty} f(x) = \infty$ and
\[
\lim_{x\to\infty} \exp\Bigl(f(x)\Bigr)~f'\Bigl(f^{-1}(\delta f(x))\Bigr) = 0, \quad \text{for any} \quad \delta \in (0,1).
\]
Examples of such functions includes: $f(x) = \beps(x) \log (x+1)$, where
$\beps(0)=1$, $\beps(x)$ is monotonically decreasing function to $0$
as $x \to \infty$; $f(x) = \sqrt{\log (x + 1)}$; $f(x) = \log \log \log (x + e^e)$, etc.

\subsection{Algorithm for Buffered Circuit Switched Network}\label{ssec:algo2}

In a buffered circuit switched network, the scheduling algorithm
decided when each of the ingress node (or queue) should request
the network for availability of resources (links) along its route
and upon positive response from the network, it acquires the
resources. Our algorithm to make such a decision at each node
is described as follows:

\begin{itemize}
\item[$\circ$] Each ingress node of a route, say $R_i$, generates
request as per a time varying Poisson process whose rate at time
$t$ is equal to $\exp(W_i(t))$.

\item[$\circ$] If the request generated by an ingress node of route,
say $R_i$, is accepted, a flow from the head of its queue leaves
the queue and acquire the resources in the network. Else, do nothing.
\end{itemize}

In above, like the algorithm for wireless network we assume that if
the request of ingress node $i$ is accepted, a new flow will acquire
resources in the network along its route. This is irrespective of
whether queue is empty or not -- if queue is empty, a {\em dummy}
flow is generated. This is merely for technical reasons.

\subsubsection{Throughput Optimality}

We describe a specific choice of weight $\bW(t)$ for which the
algorithm for circuit switched network as described above is
throughput optimal. Specifically, for route $R_i$ its weight at time
$t$ is defined as
\begin{equation}
W_i(t)  =  \max\left\{f(Q_i(\lfloor t\rfloor)),\sqrt{f({Q}_{\max}(\lfloor t\rfloor))} \right\},
\label{eq:weight2}
\end{equation}
where $f(x) = \log \log (x+e)$. The remark about distributed estimation of
$Q_{\max}(\lfloor t\rfloor))$ after \eqref{eq:weight1} applies here as well.
We state the following property of the algorithm.
\begin{theorem}\label{thm:main2}
Suppose the algorithm of Section \ref{ssec:algo2}
uses the weight as per \eqref{eq:weight2}. Then, for any
$\lamb \in \Lamb_{cs}^o$, the network Markov process
is positive Harris recurrent.
\end{theorem}

\section{Technical Preliminaries}\label{sec:prelim}

%We present some known results about the stationary distribution and
%the convergence time (or mixing time) to it for
%specific classes of finite-state Markov chains. Specifically, we consider {\em Glauber
%dynamics} (or {\em Metropolis-Hastings}) for the wireless network
%and the {\em loss networks} \cite{Kelly} for the buffered circuit switched network.
%As the reader will find, these results will play an important
%role in establishing the positive Harris recurrence of the network
%Markov process.

\subsection{Finite State Markov Chain}\label{ssec:fsmc}

Consider a discrete-time, time homogeneous Markov chain
over a finite state space $\Omega$. Let its probability
transition matrix be $P =[P_{ij}]\in \Rp^{|\Omega| \times |\Omega|}$.
If $P$ is irreducible and aperiodic, then the Markov chain
is known to have a unique stationary distribution $\pi = [\pi_i] \in \Rp^{|\Omega|}$
and it is ergodic, i.e.
$$\lim_{\tau \to\infty} P^\tau_{ji} \to \pi_i, \qquad \text{for any} \qquad i,j \in \Omega. $$
The adjoint of $P$, also known as the time-reversal of $P$, denoted by
$P^*$ is defined as follows:
\begin{eqnarray}
 \pi_iP^*_{ij} & = & \pi_j P_{ji}, \qquad \text{for any}\qquad i, j \in \Omega.\label{eq:db}
 \end{eqnarray}
By definition, $P^*$ has $\pi$ as its stationary distribution as well.
If $P = P^*$ then $P$ is called \emph{reversible} or {\em time reversible}.

Similar notions can be defined for a continuous time Markov proces over
$\Omega$. To this end, let $P(s,t) = [P_{ij}(s,t)] \in \Rp^{|\Omega| \times |\Omega|}$
denote its transition matrix over time interval $[s,t]$. The Markov
process is called time-homogeneous if $P(s,t)$ is stationary, i.e.
$P(s,t)=P(0,t-s)$ for all $0\leq s < t$ and is called reversible if
$P(s,t)$ is reversible for all $0\leq s < t$. Further, if $P(0,t)$
is irreducible and aperiodic for all $t>0$, then this time-homogeneous
reversible Markov process has a unique stationary distribution $\pi$
and it is ergodic, i.e.
$$\lim_{t \to\infty} P_{ji}(0,t) \to \pi_i, \qquad \text{for any} \qquad i,j \in \Omega. $$

\subsection{Mixing Time of Markov Chain}

Given an ergodic finite state Markov chain, the distribution at
time $\tau$ converge to the stationary distribution starting from
any initial condition as described above. We will need quantitative
bounds on the time it takes for them to reach ``close'' to their
stationary distribution. This time to reach stationarity is known
as the {\em mixing time} of the Markov chain. Here we introduce
necessary preliminaries related to this notion. We refer an interested
reader to survey papers \cite{LW, MT06}. We start with the
definition of distances between probability distributions.
\begin{definition}(Distance of measures) Given two probability
distributions $\nu$ and $\mu$ on a finite space $\Omega$, we define
the following two distances. The total variation distance, denoted
as $\norm{\nu - \mu}_{TV}$ is
$$ \norm{\nu - \mu}_{TV} = \frac12\sum_{i\in\Omega}\abs{\nu(i)-\mu(i)}. $$
The \emph{$\chi^2$ distance}, denoted as $\norm{\frac{\nu}{\mu} - 1}_{2,{\mu}}$
is
$$\norm{\frac{\nu}{\mu}-1}_{2,\mu}^2 = \norm{\nu - \mu}_{2,\frac{1}{\mu}}^2 =  \sum_{i\in\Omega}{\mu(i)\p{\frac{\nu(i)}{\mu(i)}-1}^2}. $$
More generally, for any two vectors $\bu,\bv  \in \R_+^{|\Omega|}$, we define
$$ \norm{\bv}_{2,\bu}^2 = \sum_{i \in \Omega} u_i v_i^2.$$
\end{definition}
We make note of the following relation between the two distances defined above:
using the Cauchy-Schwarz inequality, we have
\begin{equation}
\norm{\frac{\nu}{\mu}-1}_{2,\mu} \geq 2 \norm{\nu - \mu}_{TV}. \label{eq:chiTV}
\end{equation}
Next, we define a matrix norm that will be useful in determining
the rate of convergence or the mixing time of a finite-state Markov chain.
\begin{definition}[Matrix norm] %Given a matrix $P$
Consider a $|\Omega| \times |\Omega|$ non-negative valued matrix
$A \in \R_+^{|\Omega| \times |\Omega|}$ and a given vector $\bu \in \R_+^{|\Omega|}$.
Then, the matrix norm of $A$ with respect to $\bu$ is defined as follows:
$$\|A\|_\bu = \sup_{\bv : \E_{\bu}[\bv]=0 } {\frac{{\|A \bv\|}_{2,\bu}}{\|\bv\|_{2,\bu}}}, $$
where $\E_{\bu}[\bv] = \sum_{i} u_i v_i$.
\end{definition}
 It can be easily checked that the above definition of matrix norm
 satisfies the following properties.
 \begin{itemize}
 \item[{\bf P1}.] For matrices $A, B \in \R_+^{|\Omega| \times |\Omega|}$ and $\pi \in \R_+^{|\Omega|}$
 $$ \|A + B \|_\pi \leq \|A  \|_\pi + \|B \|_\pi. $$
 \item[{\bf P2}.] For matrix $A \in \R_+^{|\Omega| \times |\Omega|}$, $\pi \in \R_+^{|\Omega|}$ and $c \in \R$,
 $$ \|c A  \|_\pi = |c| \|A  \|_\pi.$$
 \item[{\bf P3}.] Let $A$ and $B$ be transition matrices of reversible Markov chains, i.e.\ $A = A^*$
 and $B = B^*$. Let both of them have $\pi$ as their unique stationary distribution. Then,
 $$ \| A B  \|_\pi \leq \|A  \|_\pi \|B  \|_\pi.$$
 \item[{\bf P4}.] Let $A$ be the transition matrix of a reversible Markov chain, i.e. $A = A^*$. Then,
$$\|A\|\leq \lambda_{\max},$$
where $\lambda_{\max}=\max\{|\lambda| \neq 1:\mbox{$\lambda$ is an eigenvalue of $A$}\}$.
 \end{itemize}
For a probability matrix $P$, we will mostly be interested in
the matrix norm of $P$ with respect to its stationary distribution $\pi$,
i.e.\ $\|P\|_\pi$. Therefore, in this paper if we use a matrix norm for
a probability matrix without mentioning the reference measure, then it is
with respect to the stationary distribution. That is, in the above example
$\|P\|$ will  mean $\|P\|_\pi$.

With these definitions, it follows that for any distribution $\mu$ on $\Omega$
\begin{equation}\label{eq:mixing}
\norm{\frac{\mu P}{\pi}-1}_{2,\pi}\leq \|P^*\|\norm{\frac{\mu}{\pi}-1}_{2,\pi},
\end{equation}
since $\E_{\pi}\left[\frac{\mu}{\pi}-1\right]=0$, where
$\frac{\mu}{\pi} = [\mu(i)/\pi(i)]$. The Markov chain of our interest, Glauber
dynamics, is reversible i.e.\ $P = P^*$.
 Therefore, for a reversible Markov chain starting with initial distribution $\mu(0)$, the
 distribution $\mu(\tau)$ at time $\tau$ is such that
 \begin{eqnarray}
 \norm{\frac{\mu(\tau)}{\pi}-1}_{2,\pi}\leq \|P\|^{\tau}\norm{\frac{\mu(0)}{\pi}-1}_{2,\pi}.\label{eq:mixing2}
 \end{eqnarray}
 Now starting from any state $i$, i.e.\ probability distribution with unit mass on
 state $i$, the initial distance $\norm{\frac{\mu(0)}{\pi}-1}_{2,\pi}$ in the worst
 case is bounded above by {$\sqrt{1/\pi_{\min}}$} where {$\pi_{\min} = \min_i \pi_i$}. Therefore,
 for any $\delta > 0$ we have $\norm{\frac{\mu(\tau)}{\pi}-1}_{2,\pi} \leq \delta$
 for any $\tau$ such that\footnote{Throughout this paper, we shall
utilize the standard order-notations: for two functions
$g, f : \R_+ \to \R_+$, $g(x) = \omega(f(x))$ means
$\liminf_{x\to\infty} g(x)/f(x) = \infty$; $g(x) = \Omega(f(x))$ means
$\liminf_{x\to\infty} g(x)/f(x) > 0$; $g(x) = \Theta(f(x))$ means
$0 < \liminf_{x\to\infty} g(x)/f(x) \leq \limsup_{x\to\infty} g(x)/f(x) < \infty$; $g(x) = O(f(x))$ means  $\limsup_{x\to\infty} g(x)/f(x) < \infty$;
$g(x) = o(f(x))$ means  $\limsup_{x\to\infty} g(x)/f(x) =0$.}
 {$$ \tau \geq \frac{\log 1/\pi_{\min} + \log 1/\delta}{\log 1/\|P\|} ~=~\Theta\left(\frac{\log 1/\pi_{\min} + \log 1/\delta}{1-\|P\|}\right).$$}

 This suggests that the ``mixing time'', i.e.\ time to reach (close to) the stationary
 distribution of the Markov chain scales inversely with $1-\|P\|$. Therefore, we will
 define the ``mixing time'' of a Markov chain with transition matrix $P$ as $1/(1-\|P\|)$.

\subsection{Glauber Dynamics \& Algorithm for Wireless Network}\label{ssec:glauber}

We will describe the relation between the algorithm for wireless network
(cf. Section \ref{ssec:algo1}) and a specific irreducible, aperiodic,
reversible Markov chain on the space of independent sets $\cI(G)$ or
schedules for wireless network with graph $G=(V,E)$. It is also known as
the Glauber dynamics, which is used by the standard
Metropolis-Hastings \cite{MRRTT53, H70} sampling mechanismthat
is described next.

\subsubsection{Glauber Dyanmics \& Its Mixing Time} We shall start off
with the definition of the Glauber dynamics followed by a useful bound
on its mixing time.

\begin{definition}[Glauber dynamics]
Consider a graph $G = (V,E)$ of $n = |V|$ nodes
with node weights $\bW = [W_i] \in \Rp^n$. The Glauber
dynamics based on weight $\bW$, denoted by $GD(\bW)$,
is a Markov chain on the space of independent sets
of $G$, $\cI(G)$. The transitions of this Markov chain
are described next. Suppose the Markov chain is
currently in the state $\sig \in \cI(G)$. Then, the next
state, say $\sig'$ is decided as follows: pick a node $i \in V$
uniformly at random and
\begin{enumerate}
\item[$\circ$] set $\sigma_j' = \sigma_j$ for $j \neq i$,
\item[$\circ$] if $\sigma_k = 0$ for all $k \in \cN(i)$, then set
$$ \sigma_i' = \begin{cases}
1  &\mbox{with probability } \frac{\exp(W_i)}{1 + \exp(W_i)}\\
0 &\mbox{otherwise,}
\end{cases}$$
\item[$\circ$] else set $\sigma_i' = 0$.
\end{enumerate}
\end{definition}
It can be verified that the Glauber dynamics $GD(\bW)$ is reversible
with stationary distribution $\pi$ given by
\begin{eqnarray}
\pi_{\sig} & \propto & \exp(\bW\cdot\sig), \qquad \text{for any} \quad \sig \in \cI(G).\label{eq:glauber}
\end{eqnarray}

Now we describe bound on the mixing time of Glauber dynamics.
\begin{lemma}\label{lem:glaumixing}
Let $P$ be the transition matrix of the Glauber dynamics $GD(\bW)$ with $n$ nodes. Then,
\beq \|P\| &\leq&  1 - \frac{1}{n^2 2^{2n+3} \exp\left(2(n+1)W_{\max}\right)},\\
 \left\|e^{n(P-I)}\right\| & \leq & 1 - \frac{1}{n 2^{2n+4} \exp(2(n+1)W_{\max})}.
\eeq
\end{lemma}
\begin{proof}
By the property {\bf P4} of the matrix norm and Cheeger's inequality \cite{C, DFK91, JS, DS, sinclair},
it is well known that $\|P\|\leq \lambda_{\max} \le 1 - \frac{\Phi^2}{2}$ where  $\Phi$ is the
conductance of $P$, defined as
$$ \Phi ~=~ \min_{S \subset \cI(G): \pi(S)\le\frac12} \frac{Q(S,S^c)}{\pi(S)\pi(S^c)},$$
where $S^c = \cI(G) \backslash S$, $Q(S,S^c) = \sum_{\sig\in S, \sig'\in S^c}{\pi(\sig)P(\sig,\sig')}.$
Now we have
  \begin{eqnarray*}
    \Phi &\ge& \min_{S\subset \cI(G)}{Q(S,S^c)} \\ & \ge  & \min_{P(\sig,\sig')\ne0} \pi(\sig)P(\sig,\sig') \\
    &\ge& \pi_{\min}\cdot\min_i \frac1n \frac1{1+\exp(W_i)} \\ & \geq & \frac1{2^n\exp(nW_{\max})}\cdot\frac1n\frac1{1+\exp(W_{\max})}\\
    &\ge& \frac1{n2^{n+1}\exp\left((n+1)W_{\max}\right)}.
  \end{eqnarray*}
Therefore
$$\norm{P} \le  1 - \frac{1}{n^2 2^{2n+3}\exp\left(2(n+1)W_{\max}\right)}.$$

Now consider $e^{n(P(\tau)-I)}$. Using properties {\bf P1}, {\bf P2} and
 {\bf P3} of matrix norm, we have:
 \begin{eqnarray}
 \left\|e^{n(P-I)}\right\| & = & \left\| e^{-n} \sum_{k=0}^\infty \frac{n^k P^k}{k!} \right\| \nonumber \\
             & \leq & e^{-n} \sum_{k=0}^\infty \frac{n^k \left\|P\right\|^k}{k!} \nonumber \\
             & = & e^{n(\|P\| - 1)} \nonumber \\
             & \leq & 1 - \frac{n(1-\|P\|)}{2}.\label{eq:dx1}
 \end{eqnarray}
 In the last inequality, we have used the fact that $\|P\| < 1$ and
 $e^{-x} \leq 1 - x/2$ for all $x \in [0,1]$. Hence, from the bound of $\|P\|$, we obtain
 \begin{eqnarray}
 \left\|e^{n(P-I)}\right\| & \leq & 1 - \frac{1}{ n 2^{2n+4} \exp(2(n+1)W_{\max})}.
 \end{eqnarray}
 This completes the proof of Lemma \ref{lem:glaumixing}.
 \end{proof}

\subsubsection{Relation to Algorithm}

Now we relate our algorithm for wireless network scheduling described
in Section \ref{ssec:algo1} with an appropriate continuous time version
of the Glauber dynamics with time-varying weights. Recall that
$\bQ(t)$ and $\sig(t)$ denote the queue-size vector and schedule
at time $t$. The algorithm changes its scheduling decision, $\sig(t)$, when a node's
exponential clock of rate $1$ ticks. Due to memoryless property of
exponential distribution and independence of clocks of all nodes,
this is equivalent to having a global exponential clock of rate $n$
and upon clock tick one of the $n$ nodes gets chosen. This node
decides its transition as explained in Section \ref{ssec:algo1}.
Thus, the effective dynamics of the algorithm upon a global clock
tick is such that the schedule $\sig(t)$ evolves exactly as per
the Glauber dynamics $GD(\bW(t))$. Here recall that $\bW(t)$ is
determined based on $\bQ(\lfloor t \rfloor)$. With abuse of notation,
let the transition matrix of this Glauber dynamics be denoted by
$GD(\bW(t))$.

Now consider any $\tau \in \Zp$. Let $\bQ(\tau), \sig(\tau)$ be the
states at time $\tau$. Then,
%\textcolor{red}
{\begin{eqnarray*}%\label{eq:fz1}
\E\left[\bdel_{\sig(\tau+1)} \, \Big| \, \bQ(\tau), \sig(\tau)\right]
& = & \sum_{k=0}^\infty \bdel_{\sig(\tau)} \Pr(\zeta = k) GD(\bW(\tau))^k,
\end{eqnarray*}
where we have used notation $\bdel_{\sig}$ for the distribution
with singleton support $\{\sig\}$ and $\zeta$ is a Poisson
random variable of mean $n$.
In above, the expectation is taken with respect to the distribution of $\sig(\tau+1)$ given $\bQ(\tau), \sig(\tau)$.} %\textcolor{red}
{Therefore, it follows that
%\iffalse
\begin{eqnarray}\label{eq:fz1}
\E\left[\bdel_{\sig(\tau+1)} \, \Big| \, \bQ(\tau), \sig(\tau)\right]
& = & \bdel_{\sig(\tau)} e^{n(GD(\bW(\tau))-I)}  \nonumber \\
& = & \bdel_{\sig(\tau)} P(\tau)
\end{eqnarray}
%\fi
where $P(\tau) \stackrel{\triangle}{=}  e^{n(GD(\bW(\tau))-I)}$.
In general, for any $\delta\in[0,1]$
\begin{eqnarray}\label{eq:fz1-1}
\E\left[\bdel_{\sig(\tau+\delta)} \, \Big| \, \bQ(\tau), \sig(\tau)\right]
& = & \bdel_{\sig(\tau)} P^{\delta}(\tau),
\end{eqnarray}
where $P^{\delta}(\tau)  \stackrel{\triangle}{=}  e^{\delta n(GD(\bW(\tau))-I)}$.}

\subsection{Loss Network \& Algorithm for Circuit Switched Network}\label{ssec:lossnet}

For the buffered circuit switched network, the Markov chain of interest
is related to the classical stochastic loss network model. This model
has been popularly utilized to study the performance of
various systems including the telephone networks, human resource
allocation, etc. (cf. see \cite{Kelly}). The stochastic loss
network model is very similar to the model of the buffered
circuit switched network with the only difference that it
does not have any buffers at the ingress nodes.

\subsubsection{Loss Network \& Its Mixing Time}

A loss network is described by a network
graph $G = (V,E)$ with capacitated links $[C_e]_{e\in E}$,
$n$ routes $\{R_i:R_i\subset E, 1\leq i\leq n\}$ and without
any buffer or queues at the ingress of each route.
For each route $R_i$, there is a dedicated exogenous,
independent Poisson arrival process with rate $\phi_i$.
Let $z_i(t)$ be number of active flows on route $i$ at
time $t$, with notation $\z(t) = [z_i(t)]$. Clearly,
$\z(t) \in \X$ due to network capacity constraints. At
time $t$ when a new exogenous flow arrives on route $R_i$,
if it can be accepted by the network, i.e. $\z(t) + e_i \in \X$,
then it is accepted with $z_i(t) \to z_i(t)+1$. Or else, it
is dropped (and hence lost forever). Each flow holds unit
amount of capacity on all links along its route for time that
is distributed as Exponential distribution with mean $1$,
independent of everything else. Upon the completion of holding time,
the flow departs and frees unit capacity on all links of its own
route.

Therefore, effectively this loss network model can be described
as a finite state Markov process with state space $\X$. Given
state $\z =[z_i] \in \X$, the possible transitions and corresponding
rates are given as
\begin{equation}
z_i ~~\to~~ \begin{cases} z_i + 1,\quad\mbox{with rate}\quad  \phi_i\quad\mbox{if}\quad \z+e_i \in \X, \\
                           z_i - 1,\quad\mbox{with rate}\quad  x_i.
             \end{cases} \label{eq:0}
\end{equation}
It can be verified that this Markov process is irreducible,
aperiodic, and time-reversible. Therefore, it is positive
recurrent (due to the finite state space) and has a unique
stationary distribution. Its stationary distribution $\pi$ is
known (cf. \cite{Kelly}) to have the following product-form:
for any $\z \in \X$,
\begin{eqnarray}
\pi_{\z} & \propto & \prod_{i=1}^n\frac{\phi_i^{z_i}}{z_i!}. \label{eq:lossnet}
\end{eqnarray}
We will be interested in the discrete-time (or {\em embedded})
version of this Markov processes, which can be defined as follows.
\begin{definition}[Loss Network]
A loss network Markov chain with capacitated graph $G=(V,E)$, capacities
$C_e, e\in E$ and $n$ routes $R_i, 1\leq i\leq n$, denoted by $\LN(\bphi)$
is a Markov chain on $\X$. The transition probabilities of this Markov
chain are described next.  Given a current state $\z \in \X$, the next
state $\z^*\in \X$ is decided by first picking a route $R_i$ uniformly
at random and performing the following:
\begin{enumerate}
\item[$\circ$] $z^*_j=z_j$ for $j\neq i$ and $z^*_i$ is decided by
$$ z^*_i = \begin{cases}
z_i+1  &\mbox{with probability } \frac{\phi_i}{\cR}\cdot\bone_{\{\z+e_i\in \X\}}\\
z_i-1 &\mbox{with probability } \frac{z_i}{\cR}\\
z_i&\mbox{otherwise.}
\end{cases},$$
\end{enumerate}
where $ \cR = \sum_i \phi_i + C_{\max}.$
\end{definition}
$\LN(\bphi)$ has the same stationary distribution as in \eqref{eq:lossnet},
and it is also irreducible, aperiodic, and reversible. Next, we state a
bound on the mixing time of the loss network Markov chain $\LN(\bphi)$
as follows.
\begin{lemma}\label{lem:lossmixing}
Let $P$ be the transition matrix of $\LN(\bphi)$ with $n$ routes. If $\bphi=\exp(\bW)$ with
\footnote{We use the following notation: given a function $g: \R \to \R$
and a $d$-dimensional vector $\bu \in \R^d$, let $g(\bu) = [g(u_i)] \in \R^d$.}
$W_i\geq 0$ for all $i$, then,
\beq
\norm{P} &\le& 1 - \frac{1}{8 n^4 C_{\max}^{2nC_{\max}+2n+2}~ \exp\left(2(nC_{\max}+1)W_{\max}\right)},\\
\norm{e^{n\cR (P - I)}}&\le&
1 - \frac{1}{16 n^3 C_{\max}^{2nC_{\max}+2n+2}~ \exp\left(2(nC_{\max}+1)W_{\max}\right)}.
\eeq
\end{lemma}
\begin{proof}
Similarly as the proof of Lemma \ref{lem:glaumixing}, a simple lower bound for the conductance $\Phi$
of $P$ is given by
\begin{eqnarray}
    \Phi &\geq & \pi_{\min} \cdot \min_{P_{\z,\z'}\ne0} P_{\z,\z'}. \label{ed0}
  \end{eqnarray}
To obtain the lower bound of $\pi_{\min}$, recall the equation \eqref{eq:lossnet},
$$\pi_{\z}  ~=~\frac1Z \prod_{i=1}^n\frac{\phi_i^{z_i}}{z_i!},$$
where $ Z = \sum_{\z \in \cX} \prod_{i=1}^n\frac{\phi_i^{z_i}}{z_i !}$, and consider the following:
\begin{eqnarray*}
Z ~ \leq ~  |\cX| \phi_{\max}^{nC_{\max}} %\nonumber \\
& \leq &  C_{\max}^n \exp(nC_{\max}W_{\max}),
\end{eqnarray*}
and
\begin{eqnarray*}
\prod_{i=1}^n\frac{\phi_i^{z_i}}{z_i !} ~ \geq ~ \frac1{\left(C_{\max}!\right)^n} %\\
& \geq &\frac1{C_{\max}^{n C_{\max}}}.
\end{eqnarray*}
By combining the above inequalities, we obtain
\beq
\pi_{\min}~\geq~ \frac1{C_{\max}^{n C_{\max}+n}\exp(nC_{\max}W_{\max})}.\label{ed1}
\eeq
On the other hand, one can bound $\min_{P_{\z,\z'}\ne0} P_{\z,\z'}$ as follows:
\beq
P_{\z,\z'}~\geq~ \frac1n \cdot \frac1{\cR}~\geq~\frac1n \cdot \frac1{n\phi_{\max}+C_{\max}}
~\geq~\frac1{2n^2C_{\max}\exp(W_{\max})},\label{ed2}
\eeq
where we use the fact that $x + y \leq 2xy$ if $x, y \geq 1$.
Now, by combining \eqref{ed1} and \eqref{ed2}, we have
$$\Phi \geq \frac1{2n^2C_{\max}^{n C_{\max}+n+1}\exp\left((nC_{\max}+1)W_{\max}\right)}.$$
Therefore, using the property {\bf P4} of the matrix norm and Cheeger's inequality,
we obtain the desired conclusion as
$$\norm{P} ~\le~ \lambda_{\max} ~\le~ 1-\frac{\Phi^2}{2}
  ~\le~ \frac1{8n^4C_{\max}^{2n C_{\max}+2n+2}\exp\left(2(nC_{\max}+1)W_{\max}\right)}. $$
Furthermore, using this bound and similar arguments in the proof of Lemma \ref{lem:glaumixing}, we have
$$\norm{e^{n\cR (P - I)}}~\le~
1 - \frac{1}{16 n^3 C_{\max}^{2nC_{\max}+2n+2}~ \exp\left(2(nC_{\max}+1)W_{\max}\right)}.$$

~~
\end{proof}

%The Markov process in the loss network in terms
%of $\LN(\bphi)$\footnote{We will be explicit in its dependence
%on $\bphi$ but not on capacity $C$ or graph $G$, as
%they will remain same throughout this paper.}:
%there is a common Poisson process of rate $n\cR$, and
%run transitions of $\LN(\bphi)$ at the instances of this
%Poisson process.

\subsubsection{Relation to Algorithm}

The scheduling algorithm for buffered circuit switched network
described in Section \ref{ssec:algo2} effectively simulates a
stochastic loss network with time-varying arrival rates $\bphi(t)$
where $\bphi_i(t) = \exp(W_i(t))$. That is, the relation of
the algorithm in Section \ref{ssec:algo2} with loss network is
similar to the relation of the algorithm in Section \ref{ssec:algo1} with
Glauber dynamics that we explained in the previous section.
To this end, for a given $\tau \in \Zp$, let $\bQ(\tau)$ and
$\z(\tau)$ be queue-size vector and active flows at time $\tau$.
With abuse of notation, let
$LN(\exp(\bW(\tau)))$ be the transition matrix of the
corresponding Loss Network with $\bW(\tau)$ dependent
on $\bQ(\tau)$. Then, for any $\delta\in[0,1]$
%\textcolor{red}
{\begin{equation}\label{eq:fz2}
%~~~& & ~~~ \\
\E\left[\bdel_{\z(\tau+\delta)} \, \Big| \, \bQ(\tau), \z(\tau)\right]
 =  \bdel_{\z(\tau)} e^{n\delta \cR(\tau)(LN(\exp(\bW(\tau)))-I)}, %\nonumber
\end{equation}}%narray}
where $\cR(\tau)=\sum_i \exp(W_i(\tau))+C_{\max}$.

\subsubsection{Positive Harris Recurrence \& Its Implication}\label{sssec:harris}

For completeness, we define the well known notion of positive
Harris recurrence (e.g.\ see \cite{dai95,dainotes}). We also state its useful implications to explain
its desirability. In this paper, we will be concerned with
discrete-time, time-homogeneous Markov process or chain
evolving over a complete, separable metric space $\sX$. Let
$\cB_\sX$ denote the Borel $\sigma$-algebra on $\sX$. We assume
that the space $\sX$ be endowed with a norm\footnote{One may
assume it to be induced by the metric of $\sX$, denoted by $d$.
For example, for any  $\bx \in \sX$, $|\bx| = d(\bzero,\bx)$
with respect to a fixed $\bzero \in \sX$.}, denoted by $|\cdot|$.
Let $X(\tau)$ denote the state of Markov chain at time
$\tau \in \Zp$.

Consider any $A \in \cB_\sX$. Define stopping time
$T_A = \inf\{\tau \geq 1 : X(\tau) \in A\}$. Then the
set $A$ is called Harris recurrent if
$$ {\Pr}_{\bx}(T_A < \infty) = 1 \qquad \mbox{for any } \bx \in \sX, $$
where $\Pr_{\bx}(\cdot) \equiv \Pr( \cdot | X(0) = \bx)$.
A Markov chain is called Harris recurrent if there exists
a $\sigma$-finite measure $\mu$ on $(\sX, \cB_\sX)$ such
that whenever $\mu(A) > 0$ for $A\in \cB_\sX$,
$A$ is Harris recurrent. It is well known that if $X$
is Harris recurrent then an essentially unique invariant
measure exists (e.g.\ see Getoor \cite{Getoor}). % [21] from Dai.
If the invariant measure is finite, then it may be normalized
to obtain a unique invariant probability measure (or stationary
probability distribution); in this
case $X$ is called positive Harris recurrent.

Now we describe a useful implication of positive Harris
recurrence. Let $\pi$ be the unique invariant (or stationary)
probability distribution of the positive Harris recurrent
Markov chain $X$. Then the following ergodic property is
satisfied: for any $x \in \sX$ and non-negative measurable
function $f: \sX \to \R_+$,
$$ \lim_{T\to\infty} \frac{1}{T} \sum_{\tau = 0}^{T-1} f(X(\tau))
\to \E_\pi[f], ~~{\Pr}_{\bx}\mbox{-almost surely}.$$
Here $\E_\pi[f] = \int f(z) \pi(z).$ Note that $\E_\pi[f]$ may
not be finite.

\subsubsection{Criteria for Positive Harris Recurrence}

Here we introduce a well known criteria for establishing
the positive Harris recurrence based on existence of
a Lyapunov function and an appropriate petit set.

We will need some definitions to begin with. Given a probability
distribution (also called sampling distribution) $a$ on $\N$,
the $a$-sampled transition matrix of the Markov chain, denoted
by $K_a$ is defined as
$$ K_a(\bx, B) = \sum_{\tau\geq 0} a(\tau)P^\tau(\bx, B), ~~\mbox{for any}~~ \bx\in \sX, ~B \in \cB_\sX.$$
Now, we define a notion of a \emph{petite} set. A non-empty set
$A \in \cB_\sX$ is called $\mu_a$-\emph{petite}
if $\mu_a$ is a non-trivial measure on $(\sX,\cB_\sX)$ and $a$ is a probability
distribution on $\N$ such that for any $\bx \in A$,
$$ K_a(\bx, \cdot) \geq \mu_a(\cdot).$$
A set is called a \emph{petite} set if it is $\mu_a$-petite for some
such non-trivial measure $\mu_a$. A known sufficient condition
to establish positive Harris recurrence of a Markov chain is to
establish positive Harris recurrence of closed petite sets
{as stated in the following lemma}. We refer
an interested reader to the book by Meyn and Tweedie \cite{MT-book}
or the recent survey by Foss and Konstantopoulos \cite{Foss-Fluid}
for details.
\begin{lemma}\label{lem:one}
  Let $B$ be a closed petite set. Suppose $B$ is Harris recurrent,
  i.e.\ $\Pr_\bx(T_B < \infty) = 1$ for any $\bx \in \sX$. Further, let
$$ \sup_{\bx \in B} \E_\bx\left[T_B\right] < \infty.$$
Then the Markov chain is positive Harris recurrent. Here $\E_{\bx}$ is
defined with respect to $\Pr_{\bx}$.
\end{lemma}
Lemma \ref{lem:one} suggests that to establish the positive
Harris recurrence of the network Markov chain, it is sufficient to find a
closed petite set that satisfies the conditions of Lemma \ref{lem:one}.
To establish positive recurrence of a closed petit set, we shall utilize
the {\em drift criteria} based on an appropriate Lyapunov function stated
in the following Lemma (cf. \cite[Theorem 1]{Foss-Fluid}).
\begin{lemma}\label{lem:two}
Let $L : \sX \to \Rp$ be a function such that
$L(\bx) \to \infty$ as $|\bx| \to \infty$. For any $\kappa > 0$,
let $B_{\kappa} = \{ \bx : L(\bx) \leq \kappa\}$. And let
there exist functions $h, g : \sX \to \Zp$ such that for any $\bx \in \sX$,
$$\E\left[L(X(g(\bx))) - L(X(0)) | X(0) = \bx \right] \leq -h(\bx),$$
that satisfy the following conditions:
\begin{itemize}
\item[(a)] $\inf_{\bx \in \sX} h(\bx) > -\infty$.
\item[(b)] $\lim\inf_{L(\bx) \to \infty} h(\bx) > 0$.
\item[(c)] $\sup_{L(\bx) \leq \gamma} g(\bx) < \infty$
for all $\gamma > 0$.
\item[(d)] $\lim\sup_{L(\bx)\to\infty} g(\bx)/h(\bx) < \infty$.
\end{itemize}
Then, there exists constant $\kappa_0 > 0$ so that for all
$\kappa_0 < \kappa$, the following holds:
\begin{eqnarray}
\E_\bx\left[ T_{B_\kappa} \right] & < &  \infty, \qquad \mbox{for any $\bx \in \sX$} \label{eq:d1a} \\
\sup_{\bx \in B_\kappa} \E_\bx\left[ T_{B_\kappa} \right] & < & \infty. \label{eq:d2a}
\end{eqnarray}
That is, $B_\kappa$ is positive recurrent.
\end{lemma}

%\section{Proof of Main Result: Theorem \ref{thm:main}}\label{sec:mainproof}
\section{Proofs of Theorems \ref{thm:main1} \& \ref{thm:main2}}\label{sec:mainproof}

This section provides proofs of Theorems \ref{thm:main1} and \ref{thm:main2}.
We shall start with necessary formalism followed by a summary of the
entire proof. This summary will utilize a series of Lemmas whose proofs
will follow.

\subsection{Network Markov Process} We describe discrete time network Markov processes
under both algorithms that we shall utilize throughout. Let $\tau\in \Zp$
be the time index. Let $\bQ(\tau) = [Q_i(\tau)]$  be the queue-size
vector at time $\tau$, $\x(\tau)$ be the schedule at time $\tau$
with $\x(\tau)=\bsigma(\tau)\in\cI(G)$ for the wireless network and
$\x(\tau)=\z(\tau)\in\X$ for the circuit switched network.
It can be checked that the tuple $X(\tau) = (\bQ(\tau), \x(\tau))$
is the Markov state of the network for both setups. Here
$X(\tau) \in \sX$ where $\sX = \R_+^n  \times {\cI(G)}$ or
$\sX = \Zp^n  \times {\X}$. Clearly, $\sX$ is a Polish space
endowed with the natural product topology. Let $\cB_{\sX}$
be the Borel $\sigma$-algebra of $\sX$ with respect
to this product topology. For any $\bx = (\bQ, \x) \in \sX$,
we define norm of $\bx$ denoted by $|\bx|$ as
$$ |\bx| = |\bQ| + |\x|, $$
where $|\bQ|$ denotes the standard $\ell_1$ norm
while $|\x|$ is defined as its index in $\{0,\dots, {|\Omega|-1}\}$, which
is assigned arbitrarily. Since $|\x|$ is always bounded,
$|\bx| \to \infty$ if and only if $|\bQ| \to \infty$. Theorems
\ref{thm:main1} and \ref{thm:main2} wish to establish that
the Markov process $X(\tau)$ is positive Harris recurrent.

\subsection{Proof Plan}

To establish the positive Harris recurrence of $X(\tau)$, we will
utilize the Lyapunov drift criteria to establish the positive recurrence
property of an appropriate petit set (cf. Lemma \ref{lem:one}). To
establish the existence of such a Lyapunov function, we shall study
properties of our randomized scheduling algorithms. Specifically,
we shall show that in a nutshell our schedule algorithms are
{\em simulating} the maximum weight scheduling algorithm with
respect to an appropriate weight, function of the queue-size.
This will lead to the desired Lyapunov function and a drift
criteria. The detailed proof of positive Harris recurrence
that follows this intuition is stated in four steps. We
briefly give an overview of these four steps.

To this end, recall that the randomized algorithms for
wireless or circuit switched network are effectively
asynchronous, continuous versions of the time-varying
$GD(\bW(t))$ or $\LN(\exp(\bW(t)))$ respectively.
Let $\pi(t)$ be the stationary distribution of the
Markov chain $GD(\bW(t))$ or $\LN(\exp(\bW(t)))$;
$\mu(t)$ be the distribution of the schedule, either
$\sig(t)$ or $\z(t)$, under our algorithm at time $t$.
In the first step, roughly speaking we argue that
the weight of schedule sampled as per the stationary distribution
$\pi(t)$ is close to the weight of maximum weight
schedule for both networks (with an appropriately
defined weight). In the second step, roughly speaking
we argue that indeed the distribution $\mu(t)$ is
close enough to that of $\pi(t)$ for all time $t$.
In the third step, using these two properties we
establish the Lyapunov drift criteria for appropriately
defined Lyapunov function (cf. Lemma \ref{lem:two}).
In the fourth and final step, we show that this implies
positive recurrence of an appropriate closed petit set.
Therefore, due to Lemma \ref{lem:one} this will imply
the positive Harris recurrence property of the network Markov
process.

\subsection{Formal Proof}

To this end, we are interested in establishing Lyapunov drift
criteria (cf. Lemma \ref{lem:two}). For this, consider Markov
process starting at time $0$ in state $X(0) = (\bQ(0), \x(0))$
and as per hypothesis of both Theorems, let
$\lamb \in (1-\beps)\Conv(\Omega)$ with some $\beps > 0$ and
$\Omega = \cI(G)$ (or $\X$).
Given this, we will go through four steps to prove positive
Harris recurrence.

\vspace{.1in}
\subsubsection{Step One}
Let $\pi(0)$ be the stationary distribution of $GD(\bW(0))$
or $LN(\exp(\bW(0)))$.  The following Lemma states that the average weight of schedule as per
$\pi(0)$ is essentially as good as that of the maximum weight
schedule with respect to weight $f(\bQ(0))$.

\begin{lemma}\label{LEM:GOODPI}
Let $\x$ be distributed over $\Omega$ as per $\pi(0)$ given $\bQ(0)$.
Then,
%if $Q_{\max}(0)$ is large enough,
\begin{equation}\label{eq:lemgoodpi}
\E_{\pi(0)}[f(\bQ(0)) \cdot \x ] ~\geq~ \left(1- \frac{\varepsilon}4\right) \left(\max_{\y\in \Omega} f(\bQ(0)) \cdot\y \right) - O(1).
\end{equation}
\end{lemma}
The proof of Lemma \ref{LEM:GOODPI} is based on the variational
characterization of distribution in the exponential form. Specifically,
we state the following proposition which is a direct adaptation of
the known results in literature (cf. \cite{GBook}).
\begin{proposition}\label{prop:goodpi}
Let $T: \Omega \to \R$ and let $\cM(\Omega)$ be space of all
distributions on $\Omega$. Define $F : \cM(\Omega) \to \R$ as
    $$F(\mu) = \E_{\mu}(T(\bx)) + H_{ER}(\mu),$$
    where $H_{ER}(\mu)$ is the standard discrete entropy of $\mu$.
    Then, $F$ is uniquely maximized by the distribution $\nu$, where
    $$ \nu_\bx = \frac{1}{Z} \exp\left(T(\bx)\right),~~\mbox{for any}~~\bx \in \Omega,$$
where $Z$ is the normalization constant (or partition function).
    Further, with respect to $\nu$, we have
    $$ \E_{\nu}[T(\bx)] \geq \left[\max_{\bx \in \cX} T(\bx)\right] - \log |\Omega|. $$
    \end{proposition}
  \begin{proof}
Observe that the definition of distribution $\nu$ implies that for any $\bx \in \Omega$,
$$T(\bx) = \log Z + \log \nu_{\bx}.$$ Using this, for any distribution
$\mu$ on $\Omega$, we obtain
    \begin{equation*}
      \begin{split}
      F(\mu) &= \sum_{\bx}\mu_{\bx}T(\bx) - \sum_{\bx}\mu_{\bx}\log \mu_{\bx}\\
      &= \sum_{\bx}\mu_{\bx}(\log Z + \log \nu_{\bx}) - \sum_{\bx}{\mu_{\bx}\log\mu_{\bx}}\\
      &= \sum_{\bx}{\mu_{\bx}\log Z} + \sum_{\bx}{\mu_{\bx}\log{\frac{\nu_{\bx}}{\mu_{\bx}}}}\\
      &= \log Z + \sum_{\bx}{\mu_{\bx}\log{\frac{\nu_{\bx}}{\mu_{\bx}}}}\\
      &\le \log Z + \log\biggl(\sum_{\bx}{\mu_{\bx}\frac{\nu_{\bx}}{\mu_{\bx}}}\biggr) \\
      &= \log Z
    \end{split}
    \end{equation*}
    with equality if and only if $\mu=\nu$. To complete other claim of proposition,
    consider  $\bx^* \in \arg\max{T(\bx)}$. {Let
$\mu$ be Dirac distribution $\mu_{\bx} = \ind{\bx=\bx^*}$.} Then, for this distribution
$$ F(\mu) = T(\bx^*).$$
But, $F(\nu) \geq F(\mu)$. Also, the maximal entropy of any distribution on
$\Omega$ is $\log |\Omega|$. Therefore,
\begin{eqnarray}
T(\bx^*) & \leq & F(\nu) \nonumber \\
         & = & \E_{\nu}[T(\bx)] + H_{ER}(\nu) \nonumber \\
         & \leq &  \E_{\nu}[T(\bx)] + \log |\Omega|. \label{eq:ed10}
\end{eqnarray}
Re-arrangement of terms in \eqref{eq:ed10} will imply the second claim
of Proposition \ref{prop:goodpi}. This completes the proof of Proposition
\ref{prop:goodpi}.
\end{proof}

\vspace{.1in}
{\em Proof of Lemma \ref{LEM:GOODPI}.} The proof is based on known observations in the context of classical Loss networks literature (cf. see \cite{Kelly}). In what follows, for simplicity
we use $\pi = \pi(0)$ for a given $\bQ = \bQ(0)$. From \eqref{eq:glauber}
and \eqref{eq:lossnet}, it follows that for both network models, the
stationary distribution $\pi$ has the following form: for any $\x \in \Omega$,
$$ \pi_{\x} \propto \prod_{i} \frac{\exp\left(W_i x_i\right)}{x_i !}
            ~=~\exp\left(\sum_i W_i x_i - \log (x_i!)\right). $$
To apply Proposition \ref{prop:goodpi}, this suggest the choice of
function $T: \cX \to \R$ as
$$ T(\x) = \sum_i W_i x_i - \log (x_i!), ~~\mbox{for any}~~\x \in \Omega.$$
Observe that for any $\x \in \Omega$, $x_i$ takes one of the finitely
many values in wireless or circuit switched network for all $i$. Therefore,
it easily follows that
$$0 ~\leq~ \sum_{i} \log (x_i!) ~\leq~ O(1), $$
where the constant may depend on $n$ and the problem parameter (e.g. $C_{\max}$
in circuit switched network).  Therefore, for any $\x \in \Omega$,
\begin{eqnarray}
T(\x) & \leq & \sum_i W_i x_i \nonumber \\
       & \leq & T(\x) + O(1). \label{eq:ed11}
\end{eqnarray}
Define $\hat{\x} = \arg\max_{\x \in \Omega} \sum_i W_i x_i$.
From \eqref{eq:ed11} and Proposition \ref{prop:goodpi},
it follows that
\begin{eqnarray}
\E_{\pi}\left[\sum_i W_i x_i\right] & \geq & \E_{\pi}\left[T(\x)\right] \nonumber \\
& \geq & \max_{\x \in \Omega} T(\x)-\log|\Omega| \nonumber \\
                     & \geq & T(\hat{\x})-\log|\Omega| \nonumber \\
                     & = & \left(\sum_i W_i \hat{x}_i\right) - O(1) - \log|\Omega|\nonumber \\
                     & = & \left(\max_{\x \in \Omega} \bW \cdot \x \right) - O(1). \label{eq:ed12}
\end{eqnarray}
From the definition of weight in both algorithms (\eqref{eq:weight1} and
\eqref{eq:weight2}) for a given $\bQ$, weight $\bW = [W_i]$ is defined as
\begin{eqnarray*}
W_i & = & \max\left\{f(Q_i),\sqrt{ f({Q}_{\max})} \right\}.
%\label{eq:weight}
\end{eqnarray*}
Define $\eta \stackrel{\Delta}{=}
\frac{\varepsilon}{4\max_{\x\in \Omega} \|x\|_1}$. To establish the
proof of Lemma \ref{LEM:GOODPI}, we will consider $Q_{\max}$
such that it is large enough satisfying
$$\eta f(Q_{\max}) ~\geq~ \sqrt{ f({Q}_{\max})}.$$
For smaller $Q_{\max}$ we do not need to argue as in that
case \eqref{eq:lemgoodpi} (due to $O(1)$ term) is straightforward.
Therefore, in the remainder we assume $Q_{\max}$ large enough.
For this large enough $Q_{\max}$, it follows that for all $i$,
\begin{eqnarray}
0 & \leq & W_i - f(Q_i) ~\leq~ \sqrt{ f({Q}_{\max})}~\leq~\eta f(Q_{\max}) \label{eq:ed13}
\end{eqnarray}
Using \eqref{eq:ed13}, for any $\x \in \Omega$,
\begin{eqnarray}
0 & \leq & \bW \cdot \x - f(\bQ) \cdot \x ~=~ (\bW - f(\bQ)) \cdot \x \nonumber \\
    & \leq & \|\x\|_1 \|\bW - f(\bQ)\|_\infty\nonumber \\
    & \leq  & \|\x\|_1 \times \eta f(Q_{\max}) \nonumber \\
    & \stackrel{(a)}{\leq} & \frac{\beps}4  f(Q_{\max}) \nonumber \\
    & \stackrel{(b)}{\leq} & \frac{\beps}4 \left(\max_{\y \in \Omega} f(\bQ) \cdot \y\right),\label{eq:L10}
\end{eqnarray}
where (a) is from our choice of $\eta=\frac{\varepsilon}{4\max_{\x\in \Omega} \|x\|_1}$.
For (b), we use the fact that the singleton set $\{i\}$, i.e.
independent set $\{i\}$ for wireless network and a single active
on route $i$ for circuit switched network, is a valid schedule. And,
for $i = \arg\max_{j} Q_j$, it has weight $f(Q_{\max})$. Therefore,
the weight of the maximum weighted schedule among all possible
schedules in $\Omega$ is at least $f(Q_{\max})$.  Finally, using
\eqref{eq:ed12} and \eqref{eq:L10} we obtain
\begin{eqnarray*}
\E_{\pi}\left[f(\bQ)\cdot \x\right]
&\geq & \E_{\pi}\left[\bW\cdot\x\right] - \frac{\varepsilon}4 \left(\max_{\y\in \Omega} f(\bQ)\cdot \y \right) \\
&\geq & \left(\max_{\y\in \Omega} \bW\cdot\y \right)
          - O(1) - \frac{\varepsilon}4 \left(\max_{\y\in \Omega} f(\bQ)\cdot \y \right) \\
&\geq & \left(\max_{\y\in \Omega} f(\bQ)\cdot \y \right)
          - O(1) - \frac{\varepsilon}4 \left(\max_{\y\in \Omega} f(\bQ)\cdot \y \right) \\
& = &  \left(1- \frac{\varepsilon}4\right) \left(\max_{\y\in \Omega} f(\bQ)\cdot \y \right)
          - O(1).
\end{eqnarray*}
This completes the proof of Lemma \ref{LEM:GOODPI}.

\vspace{.1in}
\subsubsection{Step Two}

Let $\mu(t)$ be the distribution of schedule $\x(t)$ over
$\Omega$ at time $t$, given initial state $X(0) = (\bQ(0), \x(0))$.
We wish to show that for any initial condition $\x(0) \in \Omega$,
for $t$ large (but not too large) enough, $\mu(t)$ is close to
$\pi(0)$ if $Q_{\max}(0)$ is large enough. Formal statement
is as follows.
\begin{lemma}\label{lem:adiabetic1}
For a large enough $Q_{\max}(0)$,
\beq
\norm{{\mu}(t)-\pi(0)}_{TV} < \varepsilon/4,\label{eq:adiabetic1}
\eeq
for  $t\in I=[b_1(Q_{\max}(0)), b_2(Q_{\max}(0))]$, where
$b_1,b_2$ are integer-valued functions on $\Rp$ such that
$$b_1,b_2={\sf polylog}\left(Q_{\max}(0)\right)\qquad\mbox{and}\qquad
b_2/b_1=\Theta\left(\log\left(Q_{\max}(0)\right)\right).$$
In above the constants may depend on $\varepsilon, C_{\max}$ and $n$.
\end{lemma}
The notation ${\sf polylog}(z)$ represents a positive real-valued function of $z$ that scales no
faster than a finite degree polynomial of $\log z$.

\vspace{.1in}
\noindent{\em Proof of Lemma \ref{lem:adiabetic1}.} We shall prove
this Lemma for the wireless network. The proof of buffered circuit
switch network follows in an identical manner. Hence we shall skip
it. Therefore, we shall assume $\Omega = \cI(G)$ and $\x(t) = \sig(t)$.

First, we establish the desired claim for integral times. The argument
for non integral times will follow easily as argued near the end of
this proof. For $t = \tau \in \Zp$, we have
%Recall from \eqref{eq:fz1},
{\begin{eqnarray*}
\mu(\tau+1)&=&\E\left[\bdel_{\sig(\tau+1)}\right]\\
&=&\E\left[\bdel_{\sig(\tau)}\cdot P(\tau)\right],
\end{eqnarray*}
where recall that $P(\tau) {=} e^{n(GW(\bW(\tau))-I)}$ and
the last equality follows from \eqref{eq:fz1}. Again recall that
the expectation is with respect to the joint distribution
of $\{\bQ(\tau), \sig(\tau)\}$. Hence, it follows that
\begin{eqnarray*}
\mu(\tau+1)&=&\E\left[\bdel_{\sig(\tau)}\cdot P(\tau)\right]\\
&=&\E\left[\E\left[\bdel_{\sig(\tau)}\cdot P(\tau)\,\Big|\,\bQ(\tau)\right]\right]\\
&\stackrel{(a)}{=}&\E\left[\E\left[\bdel_{\sig(\tau)}\,\Big|\,\bQ(\tau)\right]\cdot P(\tau)\right]\\
&=&\E\left[\tilde{\mu}(\tau)\cdot P(\tau)\right],
\end{eqnarray*}
where
\begin{eqnarray*}
\tilde{\mu}(\tau)=\tilde{\mu}(\bQ(\tau))\stackrel{\Delta}{=}\E\left[\bdel_{\sig(\tau)}\,\Big|\,\bQ(\tau)\right].
\end{eqnarray*}
In above the expectation is taken with respect to the conditional marginal
distribution of $\sig(\tau)$ given $\bQ(\tau)$; (a) follows since
$P(\tau)$ is a function of $\bQ(\tau)$.}  Next, we establish the relation
between $\mu(\tau)$ and $\mu(\tau+1)$.
{
\begin{eqnarray*}
\mu(\tau+1)
&=&\E\left[\tilde{\mu}(\tau)\cdot P(\tau)\right]\\
&=&\E\left[\tilde{\mu}(\tau)\cdot P(0)\right]
+\E\left[\tilde{\mu}(\tau)\cdot (P(\tau)-P(0))\right]\\
&=&\E\left[\tilde{\mu}(\tau)\right]\cdot P(0)
+e(\tau)\\
&=&\mu(\tau)\cdot P(0)
+e(\tau),
\end{eqnarray*}
where $e(\tau)  \stackrel{\Delta}{=}\E\left[\tilde{\mu}(\tau)\cdot (P(\tau)-P(0))\right]$.
Here the expectation is with respect to $\bQ(\tau)$.}
Similarly,
\begin{eqnarray}
\mu(\tau+1)
&=&\mu(\tau)\cdot P(0)
+e(\tau)\nonumber\\
&=&\left(\mu(\tau-1)\cdot P(0)
+e(\tau-1)\right)\cdot P(0)
+e(\tau)\nonumber\\
&=&\mu(\tau-1)\cdot P(0)^2
+e(\tau-1)\cdot P(0)
+e(\tau). \nonumber
\end{eqnarray}
Therefore, recursively we obtain
\begin{eqnarray}
\mu(\tau+1) &=&\mu(0)\cdot P(0)^{\tau+1}
+\sum_{s=0}^{\tau}
e(\tau-s)\cdot P(0)^s.\label{eq:relmu}
\end{eqnarray}
We will choose $b_1$ (which will depend on $Q_{\max}(0)$)
such that for $\tau \geq b_1$,
\begin{eqnarray}
\left\|\mu(0)\cdot P(0)^{\tau}-\pi(0)\right\|_{TV}&\leq& \varepsilon/8.\label{eq:defc1}
\end{eqnarray}
That is, $b_1$ is the {\em mixing time} of $P(0)$. Using inequalities
\eqref{eq:mixing2}, \eqref{eq:chiTV} and Lemma \ref{lem:glaumixing}, it
follows that
$$ b_1 \equiv b_1(Q_{\max}(0)) ~=~ {\sf polylog}\left(Q_{\max}(0)\right).$$
In above, constants may depend on $n$ and $\varepsilon$. Therefore,
from \eqref{eq:relmu} and \eqref{eq:defc1}, it suffices to show that
\begin{eqnarray}
\left\|\sum_{s=0}^{\tau-1} e(\tau-1-s)\cdot P(0)^s\right\|_1\leq \varepsilon/8,\label{eq:errsum}
\end{eqnarray}
for $\tau \in I=[b_1,b_2]$ with an appropriate choice of $b_2 = b_2(Q_{\max}(0))$.
To this end, we choose
$$ b_2 \equiv b_2(Q_{\max}(0)) ~=~ \lceil b_1 \log(Q_{\max}(0)) \rceil.$$
Thus, $b_2(Q_{\max}(0)) ~=~ {\sf polylog}\left(Q_{\max}(0)\right)$ as well.
With this choice of $b_2$, we obtain the following bound on $e(\tau)$ to conclude
\eqref{eq:errsum}.
{
\begin{eqnarray}
\|e(\tau)\|_1&=&\|\E\left[\tilde{\mu}(\tau)\cdot (P(\tau)-P(0))\right]\|_1\nonumber\\
&\leq&\E\left[\|\tilde{\mu}(\tau)\cdot (P(\tau)-P(0))\|_1\right]\nonumber\\
&\stackrel{(a)}{\leq}&O\left(\E\left[\|P(\tau)-P(0)\|_{\infty}\right]\right)\nonumber\\
%&\leq&O\left(\E_{\bQ(t)}\left[\sum_{s=0}^{t-1}\|P(s+1)-P(s)\|_{\infty}\right]\right),
&\stackrel{(b)}{=}&O\left(\E\left[\left\|GW(\bW(\tau))-GW(\bW(0))\right\|_{\infty}\right]\right)\nonumber\\
&\stackrel{(c)}{=}&O\left(\E\left[\max_i \left|\frac1{1+\exp(W_i(\tau))}-\frac1{1+\exp(W_i(0))}\right|\right]\right)\nonumber\\
&\stackrel{(d)}{=}&O\left(\E\left[\max_i\left|W_i(\tau)-W_i(0)\right|\right]\right)\nonumber\\
&\stackrel{(e)}{=}&O\left(\max_i\,\E\left[\left|W_i(\tau)-W_i(0)\right|\right]\right)\label{eq:boundet}.
\end{eqnarray}
}
In above, (a) follows from the standard norm inequality and the
fact that $\|\tilde{\mu}(\tau)\|_1 = 1$, (b) follows from Lemma \ref{lem:last}
in Appendix,
(c) follows directly from the definition of transition matrix $GD(\bW)$, (d)
follows from 1-Lipschitz\footnote{A function $f:\R\rightarrow\R$ is $k$-Lipschitz if $|f(s)-f(t)|\leq k|s-t|$ for all $s,t\in\R$.} property of function $1/(1+e^x)$ and (e) follows from
the fact that vector $\bW(\tau)$ being $O(1)$ dimensional\footnote{We note here that the $O(\cdot)$ notation means existences of constants that do not depend scaling quantities such as time $\tau$ and $\bQ(0)$; however it may depend on the fixed system parameters such as number of queues. The use of this terminology is to retain the clarity of exposition.}.

Next, we will show that for
all $i$ and $\tau\leq b_2$,
{
\begin{eqnarray}
\E\left[\left|W_i(\tau)-W_i(0)\right|\right]
&=&O\left(\frac1{{\sf superpolylog}\left(Q_{\max}(0)\right)}\right),
\label{eq:boundet2}
\end{eqnarray}
}
the notation ${\sf superpolylog}(z)$ represents a positive
real-valued function of $z$ that scales
faster than any finite degree polynomial of $\log z$.
This is enough to conclude \eqref{eq:errsum} (hence complete the
proof of Lemma \ref{lem:adiabetic1}) since
\begin{eqnarray*}
\left\|\sum_{s=0}^{\tau-1}e(\tau-1-s)\cdot P(0)^s\right\|_1&\leq&
\sum_{s=0}^{\tau-1}\left\|e(\tau-1-s)\cdot P(0)^s\right\|_1\\
&=&\sum_{s=0}^{\tau-1}O\left(\left\|e(\tau-1-s)\right\|_1\right)\\
&\stackrel{(a)}{=}&O\left(\frac{\tau}{{\sf superpolylog}\left(Q_{\max}(0)\right)}\right)\\
&\stackrel{(b)}{\leq}&\frac{\varepsilon}4,
\end{eqnarray*}
where we use \eqref{eq:boundet} and \eqref{eq:boundet2} to obtain (a),
(b) holds for large enough $Q_{\max}(0)$ and $\tau\leq b_2={\sf polylog}\left(Q_{\max}(0)\right)$.

\vspace{.1in}
Now to complete the proof, we only need to establish \eqref{eq:boundet2}.
This is the step that utilizes `slowly varying' property of function
$f(x) = \log\log (x+e)$. First, we provide an intuitive sketch of the
argument. Somewhat involved details will be follow. To explain the
intuition behind \eqref{eq:boundet2}, let us consider a simpler
situation where $i$ is such that $Q_i(0) = Q_{\max}(0)$ and
$f(Q_i(\tau)) > \sqrt{f(Q_{\max}(\tau))}$ for a given
$\tau \in [0,b_2]$. That is, let $W_i(\tau) = f(Q_i(\tau))$. Now,
consider following sequence of inequalities:
%\footnote{The constants captured through $O(\cdot)$ terms do not scale with $\tau$ or $\bQ(0)$; however they may depend on system parameters such as the number of queues. The use of $O(\cdot)$ notation is primarily for clarity of exposition.}:
%\iffalse
\begin{eqnarray}
|W_i(\tau) - W_i(0)| & = & |f(Q_i(\tau) - f(Q_i(0))| \nonumber \\
 & \stackrel{(a)}{\leq} & f'(\zeta) |Q_i(\tau) - Q_i(0)|, \qquad \text{for some $\zeta$ around $Q_i(0)$} \nonumber \\
 & \stackrel{(b)}{\leq}  & f'(\min\{Q_i(\tau), Q_i(0)\}) O(\tau) \nonumber \\
 & \stackrel{(c)}{\leq} & f'(Q_i(0)-O(\tau)) O(\tau) \nonumber \\
 & \stackrel{(d)}{=} & O\left(\frac{\tau}{Q_i(0)}\right). \label{eq:fz3}
\end{eqnarray}
%\fi
In above, (a) follows from the mean value theorem; (b) follows from
monotonicity of $f'$ and Lipschitz property
of $Q_i(\cdot)$ (as a function of $\tau$)
%\footnote{Random variables $\{X(t)\in \R:t\geq 0\}$ is $k$-Lipschitz if $|X(t)-X(s)|\leq k|s-t|$
%or all $s,t$ with probability 1.}
-- which holds deterministically for wireless network
and probabilistically for circuit switched network; (c) uses the
same Lipschitz property; and (d) uses the fact that $\tau \leq b_2$ and
$b_2 = {\sf polylog}(Q_{\max}(0))$, $Q_{\max}(0) = Q_i(0)$.
Therefore, effectively the bound of \eqref{eq:fz3} is
$O(1/{\sf superpolylog}(Q_{\max}(0))$.

The above explains the gist of the argument that is to follow. However,
to make it precise, we will need to provide lots more details. Toward this,
we consider the following two cases: (i) $f(Q_i(0))\geq \sqrt{f(Q_{\max}(0))}$,
and (ii) $f(Q_i(0))< \sqrt{f(Q_{\max}(0))}$. In what follows, we provide
detailed arguments for (i). The arguments for case (ii) are similar in
spirit and will be provided later in the proof.

\vspace{.1in}
\noindent{\em Case (i):} Consider an $i$ such that $f(Q_i(0))\geq \sqrt{f(Q_{\max}(0))}$.
Then,
{
\begin{eqnarray}
& & \E\left[\left|W_i(\tau)-W_i(0)\right|\right]\nonumber \\
& & ~ = \E\left[\left|W_i(\tau)-f(Q_i(0))\right|\right]\nonumber\\
& & ~ =\E\left[\left|f(Q_i(\tau))-f(Q_i(0))\right|\cdot
\bind_{\left\{f(Q_i(\tau))\geq \sqrt{f(Q_{\max}(\tau))}\right\}}\right]\nonumber\\
&& \qquad + ~\E\left[\left|\sqrt{f(Q_{\max}(\tau))}-f(Q_i(0))\right|
\cdot \bind_{\left\{f(Q_i(\tau))< \sqrt{f(Q_{\max}(\tau))}\right\}}\right],\label{eq:js1}
\end{eqnarray}}
where each equality follows from \eqref{eq:weight1}.
The first term in \eqref{eq:js1} can be bounded as follows
{
\begin{eqnarray}
&&\E\left[\left|f(Q_i(\tau))-f(Q_i(0))\right|\cdot
\bind_{\{f(Q_i(\tau))\geq \sqrt{f(Q_{\max}(\tau))}\}}\right]\nonumber\\
& & \leq\E\left[\left|f(Q_i(\tau))-f(Q_i(0))\right|\right]\nonumber\\
& & \stackrel{(o)}{\leq} \E\left[f^{\prime}\left(\min\{Q_i(\tau),Q_i(0)\}\right)|Q_i(\tau)-Q_i(0)|\right]\nonumber\\
& & \stackrel{(a)}{\leq} \sqrt{\E\left[f^{\prime}\left(\min\{Q_i(\tau),Q_i(0)\}\right)^2\right]}\cdot
\sqrt{\E\left[(Q_i(\tau)-Q_i(0))^2\right]}\nonumber\\
& & \stackrel{(b)}{\leq} \sqrt{f^{\prime}\left(\frac{Q_i(0)}2\right)^2+\Theta\left(\frac{\tau}{Q_i(0)}\right)} \cdot O(\tau)\nonumber\\
& & \stackrel{(c)}{\leq} \sqrt{f^{\prime}\left(\frac12f^{-1}\left(\sqrt{f(Q_{\max}(0))}\right)\right)^2
+\Theta\left(\frac{\tau}{f^{-1}\left(\sqrt{f(Q_{\max}(0))}\right)}\right)} \cdot O(\tau)\nonumber\\
& & \stackrel{(d)}{=} O\left(\frac1{{\sf superpolylog}\left(Q_{\max}(0)\right)}\right).\label{eq:js2}
\end{eqnarray}}
In above, (o) follows from concavity of $f$. For (a) we use the standard
Cauchy-Schwarz inequality $\E[XY]\leq \sqrt{\E[X^2]}\sqrt{\E[Y^2]}$.
For (b), note that given $Q_i(0)$, $\E[[Q_i(0) - Q_i(\tau)]^2] = O(\tau^2)$ for
both network models -- for wireless network, it is deterministically true due to
Lipschitz property of $\bQ(\cdot)$; for circuit switched network, it is due to
the fact that the arrival as well as (the overall) departure processes
are bounded rate Poisson processes.  Given this, using Markov's inequality
it follows that $$\Pr\left(\min\{Q_i(\tau), Q_i(0)\} \leq \frac{Q_i(0)}2\right) ~=~ O\left(\frac{\tau}{Q_i(0)}\right).$$
Finally, using the fact that $\sup_{y \in \Rp} f'(y) = O(1)$, we obtain (b).
Now (c) follows from the condition of $Q_i(0)$ that
$f(Q_i(0))\geq \sqrt{f(Q_{\max}(0))}$. And, (d) is implied by
$\tau\leq b_2={\sf polylog}(Q_{\max}(0))$, $f(x)=\log\log(x+e)$.

Next, we bound the second term in \eqref{eq:js1}. We will use notation
$$A(\tau) = \left\{f(Q_i(\tau))< \sqrt{f(Q_{\max}(\tau))} \,\&\,
\sqrt{f(Q_{\max}(\tau))}\geq f(Q_i(0))\right\},  $$
$$ B(\tau) = \left\{f(Q_i(\tau))< \sqrt{f(Q_{\max}(\tau))} \,\&\,
\sqrt{f(Q_{\max}(\tau))} <  f(Q_i(0))\right\}.$$
Then,
{
\begin{eqnarray}
&& \E\left[\left|\sqrt{f(Q_{\max}(\tau))}-f(Q_i(0))\right|
\cdot \bind_{\left\{f(Q_i(\tau))< \sqrt{f(Q_{\max}(\tau))}\right\}}\right]\nonumber\\
& & \quad =\E\left[\left(\sqrt{f(Q_{\max}(\tau))}-f(Q_i(0))\right)
\cdot \bind_{A(\tau)}\right] \nonumber \\ & & \quad\qquad + \quad \E\left[\left(f(Q_i(0))-\sqrt{f(Q_{\max}(\tau))}\right)
\cdot \bind_{B(\tau)}\right] \nonumber \\
& & \quad \stackrel{(a)}{\leq} \E\left[\left(\sqrt{f(Q_{\max}(\tau))}-\sqrt{f(Q_{\max}(0))}\right)
\cdot \bind_{A(\tau)}\right] \nonumber\\
& & \quad\qquad + \quad  \E\left[\left(f(Q_i(0))-f(Q_i(\tau))\right)
\cdot \bind_{B(\tau)}\right]\nonumber\\
%& & \quad\stackrel{(b)}{\leq} \E_{\bQ(\tau)}\left[\sqrt{f(Q_{\max}(\tau))}-\sqrt{f(Q_{\max}(0))}\right]
%+\E_{\bQ(\tau)}\left[f(Q_i(0))-f(Q_i(\tau))\right]\nonumber\\
& & \quad \stackrel{(b)}{\leq} \E\left[|f(Q_{\max}(\tau))-f(Q_{\max}(0))|\right]
+\E\left[|f(Q_i(0))-f(Q_i(\tau))|\right]\nonumber\\
& & \quad =  O\left(\frac1{{\sf superpolylog}\left(Q_{\max}(0)\right)}\right).
\label{eq:js3}
\end{eqnarray}}
In above, (a) follows because we are considering case (i) with $f(Q_i(0)) \geq
\sqrt{f(Q_{\max}(0))}$ and definition of event $B(\tau)$; (b) follows from
1-Lipschitz property of $\sqrt{\cdot}$ function and appropriate removal of
indicator random variables. For the final conclusion, we observe that
the arguments used to establish \eqref{eq:js2} imply the
$O(1/{\sf superpolylog}(Q_{\max}(0)))$ bound on both the terms
in very similar manner: for the term corresponding
to $|f(Q_{\max}(\tau))-f(Q_{\max}(0))|$, one has to adapt arguments
of \eqref{eq:js2} by essentially replacing the index $i$ by $\max$.
This concludes the proof of \eqref{eq:boundet2}  for case
(i) of $f(Q_i(0)) \geq \sqrt{f(Q_{\max}(0))}$.

\vspace{.1in}
\noindent{\em Case (ii):} Now consider $i$ such that
$f(Q_i(0)) < \sqrt{f(Q_{\max}(0))}$.  Then,
%like \eqref{eq:js1} we have
{
\begin{eqnarray}
& & \E\left[\left|W_i(\tau)-W_i(0)\right|\right] %\nonumber \\
% & &
~ = \E\left[\left|W_i(\tau)-\sqrt{f(Q_{\max}(0))}\right|\right]\nonumber\\
& & ~ =\E\left[\left|f(Q_i(\tau))-\sqrt{f(Q_{\max}(0))}\right|\cdot
\bind_{\left\{f(Q_i(\tau))\geq \sqrt{f(Q_{\max}(\tau))}\right\}}\right]\nonumber\\
&& \qquad + ~\E\left[\left|\sqrt{f(Q_{\max}(\tau))}-\sqrt{f(Q_{\max}(0))}\right|
\cdot \bind_{\left\{f(Q_i(\tau))< \sqrt{f(Q_{\max}(\tau))}\right\}}\right].\label{eq:js1a}
\end{eqnarray}}
First observe that by 1-Lipschitz property of $\sqrt{\cdot}$ function, the
second term can be bounded as (similar to \eqref{eq:js3})
{
\begin{eqnarray}
& & \E\left[\left|\sqrt{f(Q_{\max}(\tau))}-\sqrt{f(Q_{\max}(0))}\right|
\cdot \bind_{\left\{f(Q_i(\tau))< \sqrt{f(Q_{\max}(\tau))}\right\}}\right] \nonumber \\
& & \quad \leq  \E\left[\left|{f(Q_{\max}(\tau))}-{f(Q_{\max}(0))}\right|\right] \nonumber \\
& & \quad =  O\left(\frac1{{\sf superpolylog}\left(Q_{\max}(0)\right)}\right).
\end{eqnarray}}
Therefore, we are left with proving the first term of \eqref{eq:js1a}. We
will follow similar line of arguments as those used for \eqref{eq:js3}.
Define
$$A'(\tau) = \left\{f(Q_i(\tau)) \geq \sqrt{f(Q_{\max}(\tau))} \,\&\,
\sqrt{f(Q_{\max}(0))}\geq f(Q_i(\tau))\right\},  $$
$$ B'(\tau) = \left\{f(Q_i(\tau)) \geq \sqrt{f(Q_{\max}(\tau))} \,\&\,
\sqrt{f(Q_{\max}(0))} <  f(Q_i(\tau))\right\}.$$
Then,
{
\begin{eqnarray}
&& \E\left[\left|f(Q_i(\tau))-\sqrt{f(Q_{\max}(0))}\right|
\cdot \bind_{\left\{f(Q_i(\tau))\geq \sqrt{f(Q_{\max}(\tau))}\right\}}\right]\nonumber\\
& & \quad =\E\left[\left(\sqrt{f(Q_{\max}(0))}-f(Q_i(\tau))\right)
\cdot \bind_{A(\tau)}\right] \nonumber \\ & & \quad\qquad + \quad \E\left[\left(f(Q_i(\tau))-\sqrt{f(Q_{\max}(0))}\right)
\cdot \bind_{B(\tau)}\right] \nonumber \\
& & \quad \stackrel{(a)}{\leq} \E\left[\left(\sqrt{f(Q_{\max}(0))}-\sqrt{f(Q_{\max}(\tau))}\right)
\cdot \bind_{A(\tau)}\right] \nonumber\\
& & \quad\qquad + \quad  \E\left[\left(f(Q_i(\tau))-\sqrt{f(Q_{\max}(0))}\right)
\cdot \bind_{B(\tau)}\right]\nonumber\\
& & \quad \stackrel{(b)}{\leq} O\left(\frac1{{\sf superpolylog}\left(Q_{\max}(0)\right)}\right) \nonumber \\
& & \qquad \qquad + \quad \E\left[\left(f(Q_i(\tau))-\sqrt{f(Q_{\max}(0))}\right)
\cdot \bind_{B(\tau)}\right]. \nonumber \\
\label{eq:js3a}
\end{eqnarray}}
In above, (a) follows because we are considering case (i) with $f(Q_i(\tau)) \geq
\sqrt{f(Q_{\max}(\tau))}$ and definition of event $B(\tau)$; (b) follows from
1-Lipschitz property of $\sqrt{\cdot}$ function and appropriate removal of
indicator random variables as follows:
{
\begin{eqnarray}
& & \E\left[\left(\sqrt{f(Q_{\max}(0))}-\sqrt{f(Q_{\max}(\tau))}\right)\cdot \bind_{A(\tau)}\right] \nonumber \\
& & \qquad \leq   \E\left[|f(Q_{\max}(\tau))-f(Q_{\max}(0))|\right] \nonumber \\
& & \qquad =  O\left(\frac1{{\sf superpolylog}\left(Q_{\max}(0)\right)}\right).
\end{eqnarray}}
Finally, to complete the proof of case (ii) using \eqref{eq:js1a}, we wish
to establish
{
\begin{eqnarray}
\E\left[\left(f(Q_i(\tau))-\sqrt{f(Q_{\max}(0))}\right)
\cdot \bind_{B(\tau)}\right] & = &  O\left(\frac1{{\sf superpolylog}\left(Q_{\max}(0)\right)}\right).\nonumber  \\
\label{eq:js3b}
\end{eqnarray}}
Now suppose $x \in \Rp$ be such that $f(x) = \sqrt{f(Q_{\max}(0)}$. Then,
%\iffalse
{
\begin{eqnarray}
& & \E\left[\left(f(Q_i(\tau))-\sqrt{f(Q_{\max}(0))}\right)\cdot \bind_{B(\tau)}\right] \nonumber \\
& & \quad =  \E\left[\left(f(Q_i(\tau))- f(x)\right)\cdot \bind_{B(\tau)}\right] \nonumber \\
& & \quad \stackrel{(a)}{\leq} \E\left[f'(x) (Q_i(\tau) - x) \cdot \bind_{B(\tau)}\right] \nonumber \\
& & \quad = f'(x)~ \E\left[ (Q_i(\tau) - x) \cdot \bind_{B(\tau)}\right] \nonumber \\
& & \quad \stackrel{(b)}{\leq} f'(x)~ \E\left[ (Q_i(\tau) - Q_i(0)) \cdot \bind_{B(\tau)}\right] \nonumber \\
& & \quad \leq f'(x)~ \E\left[ |Q_i(\tau) - Q_i(0)|\right] \nonumber \\
& & \quad \stackrel{(c)}{=} f'(x)~ O\left({\tau}\right) \nonumber \\
& & \quad \stackrel{(d)}{=}  O\left(\frac1{{\sf superpolylog}\left(Q_{\max}(0)\right)}\right).\label{eq:js3c}
\end{eqnarray}}
%\fi
In above, (a) follows from concavity of $f$; (b) from $Q_i(0) \leq x$ and
$Q_i(\tau) \geq x$ implied by case (ii) and $B'(\tau)$ respectively; (c) follows
from arguments used earlier that for any $i$, $\E[(Q_i(\tau)-Q_i(0))^2] = O(\tau^2)$;
(d) follows from $\tau \leq b_2 = {\sf polylog}\left(Q_{\max}(0)\right)$ and
$$f'(x) = O\left(\frac1{{\sf superpolylog}\left(Q_{\max}(0)\right)}\right).$$
This complete the proof of \eqref{eq:boundet2} for both cases and
the proof of Lemma \ref{lem:adiabetic1} for integral time steps. A final
remark validity of this result about non-integral times is in order.

Consider $t \in I$ and $t \notin \Zp$. Let $\tau = \lfloor t \rfloor$ and $t = \tau + \delta$
for $\delta \in (0,1)$. Then, it follows that (using formal definition
$P^\delta$ as in \eqref{eq:fz1-1})
{
\begin{eqnarray}
\mu(t)~=~ \mu(\tau+\delta) & = & \mu(\tau) P^{\delta}(0) + \E\left[\tilde{\mu}(\tau) (P^\delta(\tau)-P^\delta(0))\right] \nonumber \\
& = & \mu(0)P(0)^\tau P^\delta(0) + e(\tau+\delta). \label{eq:fz5}
\end{eqnarray}}
Now it can be checked that $P^\delta(0)$ is a probability matrix and
has $\pi(0)$ as its stationary distribution for any $\delta > 0$; and
we have argued that for $\tau$ large enough $\mu(0)P(0)^\tau$ is close
to $\pi(0)$. Therefore, $\mu(0)P(0)^\tau P^\delta(0)$ is also equally
close to $\pi(0)$. For $e(\tau+\delta)$, it can be easily argued that
the bound obtained in \eqref{eq:boundet} for $e(\tau+1)$ will dominate the bound
for $e(\tau+\delta)$. Therefore, the statement of Lemma holds for
any non-integral $t$ as well. This complete the proof of
Lemma \ref{lem:adiabetic1}.

\vspace{.1in}
\subsubsection{Step Three: Wireless Network}\label{sssec:step3wireless}

In this section, we prove Lemma \ref{lem:two} for the wireless network model.
For Markov process $X(t) = (\bQ(t), \sig(t))$, we consider Lyapunov function
$$L(X(t)) = \sum_{i} F(Q_i(t)), $$
where $F(x) = \int_{0}^x f(y)~ dy$ and recall that $f(x) =
\log \log (x+e)$.
For this Lyapunov function, it suffices to find appropriate functions
$h$ and $g$ as per Lemma \ref{lem:two} for a large enough $Q_{\max}(0)$.
Therefore, we assume that $Q_{\max}(0)$ is large enough so that it
satisfies the conditions of Lemma \ref{lem:adiabetic1}.  To this end,
from Lemma \ref{lem:adiabetic1}, we have that for $t \in I$,
\begin{eqnarray*}
    \abs{\E_{\pi(0)}[f(\bQ(0))\cdot\sig] - \E_{{\mu}(t)}[f(\bQ(0))\cdot\sig]}
        &\leq& \frac{\varepsilon}4 \left(\max_{\rhO\in \cI(G)}{f(\bQ(0))\cdot\rhO}\right).
  \end{eqnarray*}
Thus from Lemma \ref{LEM:GOODPI}, it follows that
\begin{equation}
    \E_{{\mu}(t)}[f(\bQ(0))\cdot\sig] %&\ge \E_{\pi(t)}[f(\bQ(t))\cdot\sig] - \varepsilon \left(\max_{\rhO\in \cI(G)} f(\bQ(t))\cdot\rhO \right) \\
    ~\ge~ \left(1-\frac{\beps}2\right)\left(\max_{\rhO \in \cI(G)} f(\bQ(0))\cdot\rhO  \right) - O(1).\label{eq:L47}
\end{equation}
Now we can bound the difference between $L(X(\tau+1))$ and $L(X(\tau))$ as follows.
\begin{eqnarray}
& &   L(X(\tau+1)) - L(X(\tau))  =  (F(\bQ(\tau+1)) - F(\bQ(\tau))) \cdot \bone \nonumber\\
    & & \quad \leq~ f(\bQ(\tau+1)) \cdot (\bQ(\tau+1)-\bQ(\tau)), \nonumber\\%\qquad \qquad \text{since $F$ is convex}, \nonumber\\
   &  & \quad {\leq}~ f(\bQ(\tau)) \cdot (\bQ(\tau+1)-\bQ(\tau))+n,\nonumber%\\
    \end{eqnarray}
where the first inequality is from the convexity of $F$ and
the last inequality follows from the fact that $f(Q)$ is $1$-Lipschitz. Therefore,
\begin{eqnarray}
& &   L(X(\tau+1)) - L(X(\tau))  =  (F(\bQ(\tau+1)) - F(\bQ(\tau))) \cdot \bone \nonumber\\
    & & \quad \leq  f(\bQ(\tau)) \cdot \left(A(\tau, \tau+1) - \int^{\tau+1}_\tau \sig(y)\bone_{\{Q_i(y) > 0\}}\, dy\right)+n \nonumber\\
    & & \quad \stackrel{(a)}{\leq}  f(\bQ(\tau)) \cdot A(\tau, \tau+1)-\int^{\tau+1}_\tau f(\bQ(y)) \cdot \sig(y) \bone_{\{Q_i(y) > 0\}}  \,dy+2n\nonumber\\
   & & \quad =  f(\bQ(\tau)) \cdot A(\tau, \tau+1)-\int^{\tau+1}_\tau f(\bQ(y)) \cdot \sig(y) \,dy+2n,\label{eq:L7}
    \end{eqnarray}
where again (a) follows from the fact that $f(Q)$ is $1$-Lipschitz. Given
initial state $X(0) = \bx$, taking the expectation of \eqref{eq:L7} for $\tau, \tau +1 \in I$,
\begin{eqnarray*}
 & &  \E_\bx[L(X(\tau+1)) - L(X(\tau))] \nonumber \\
 & & \qquad \leq  \E_\bx[f(\bQ(\tau)) \cdot A(\tau, \tau+1)]
      -\int^{\tau+1}_\tau \E_\bx[f(\bQ(y)) \cdot \sig(y)] \,dy ~+~ 2n \\
& & \qquad =  \E_\bx[f(\bQ(\tau)) \cdot \lamb] -\int^{\tau+1}_\tau
\E_{\bx}[f(\bQ(y)) \cdot \sig(y)] \,dy ~+ ~2n,
 \end{eqnarray*}
  where the last equality follows from the independence between $\bQ(\tau)$ and
  $A(\tau,\tau+1)$ (recall, Bernoulli arrival process). Therefore,
\begin{eqnarray*}
&& \E_\bx[L(X(\tau+1))- L(X(\tau))]  \nonumber \\
& \leq&   \E_\bx[f(\bQ(\tau)) \cdot \lamb]
   -\int^{\tau+1}_\tau \E_{\bx}[f(\bQ(0)) \cdot \sig(y)] \,dy \\
   & &\qquad\qquad   -\int^{\tau+1}_\tau \E_{\bx}[\left(f(\bQ(y))-f(\bQ(0))\right) \cdot \sig(y)] \,dy+2n\\
      &\stackrel{(a)}{\leq}&   f(\bQ(0)+\tau\cdot \bone) \cdot \lamb%\\
 %     & &\qquad-
      -\int^{\tau+1}_\tau \E_{\bx}[f(\bQ(0)) \cdot \sig(y)] \,dy\\
      & &\qquad-\int^{\tau+1}_\tau \left(f(\bQ(0)-y\cdot \bone)-f(\bQ(0))\right) \cdot \bone \,dy+ 2n \\
      &\stackrel{(b)}{\leq}&   f(\bQ(0)) \cdot \lamb +f(\tau \cdot \bone) \cdot \lamb
      -\left(1-\frac{\beps}2\right)\left(\max_{\rhO\in\cI(G)}f(\bQ(0)) \cdot \rhO\right) \qquad \qquad \qquad \qquad\\
      && \qquad+\int^{\tau+1}_\tau f(y\cdot \bone) \cdot \bone \,dy+O(1)\\
            &{\leq}&   f(\bQ(0)) \cdot \lamb
      -\left(1-\frac{\beps}2\right)\left(\max_{\rhO\in\cI(G)}f(\bQ(0)) \cdot \rhO\right)+ 2n f(\tau+1)+O(1).
      \end{eqnarray*}
In above, (a) uses Lipschitz property of $\bQ(\cdot)$ (as a function of
$\tau$); (b) follows from \eqref{eq:L47} and the inequality that for
$f(x) = \log\log (x+e)$, $f(x)+f(y)+ \log 2 \geq f(x+y)$ for all
$x,y\in \R_+$. The $O(1)$ term is constant, dependent on $n$, and captures
the constant from \eqref{eq:L47}.

\noindent Now since $\lamb \in (1-\beps)\Conv(\cI(G))$, we obtain
\begin{eqnarray*}
 & & \E_\bx[L(X(\tau+1))- L(X(\tau))] \\
      & & \qquad \leq   -\frac{\beps}2\left(\max_{\rhO\in\cI(G)}f(\bQ(0)) \cdot \rhO\right)+2n\,f(\tau+1)+O(1)\\
      & & \qquad \leq    -\frac{\beps}2f(Q_{\max}(0))+2n\,f(\tau+1)+O(1).
  \end{eqnarray*}
Therefore, summing $\tau$ from  $b_1 = b_1(Q_{\max}(0))$ to $b_2= b_2(Q_{\max}(0))$, we have
\begin{eqnarray}
 & &     \E_\bx\left[L(X(b_2)) - L(X(b_1))\right ] \nonumber \\
    & & \qquad \le - \frac{\beps}{2}(b_2-b_1)f(Q_{\max}(0))+2n\sum_{\tau=b_1}^{b_2-1}f(\tau+1) +O(b_2-b_1)\nonumber\\
    & & \qquad \le - \frac{\beps}{2}(b_2-b_1)f(Q_{\max}(0))+2n(b_2-b_1)f(b_2) +O(b_2-b_1).
   \end{eqnarray}
Thus, we obtain
\begin{eqnarray}
&&     \E_\bx\left[L(X(b_2)) - L(X(0))\right ]\nonumber\\
     &=&     \E_\bx\left[L(X(b_1)) - L(X(0))\right ]+
          \E_\bx\left[L(X(b_2)) - L(X(b_1))\right ]\nonumber\\
    &\stackrel{(a)}{\le}& \E_\bx\left[f(\bQ(b_1))\cdot (\bQ(b_1)-\bQ(0))\right]
    -\frac{\beps}{2}(b_2-b_1)f(Q_{\max}(0)) \nonumber\\
    & & \qquad \qquad +2n\sum_{\tau=b_1}^{b_2-1}f(\tau+1)      +O(b_2-b_1)\nonumber\\
    &\stackrel{(b)}{\le}& nb_1\,f(Q_{\max}(0)+b_1))
    - \frac{\beps}{2}(b_2-b_1)f(Q_{\max}(0))\nonumber\\
    & & \qquad \qquad +2n(b_2-b_1)f(b_2) +O(b_2-b_1),\label{eq:corenegative}
   \end{eqnarray}
   where (a) follows from the convexity of $L$ and (b) is due to the $1$-Lipschitz property of $\bQ$.
Now if we choose $g(\bx) = b_2$ and
$$h(\bx)=-n b_1\,f(Q_{\max}(0)+b_1))+\frac{\beps}{2}(b_2-b_1)f(Q_{\max}(0))-2n(b_2-b_1)f(b_2) -O(b_2-b_1),$$
the desired inequality follows:
\begin{eqnarray*}
\E_{\bx}\left[L(X(g(\bx))) - L(X(0))\right] & \leq & - h(\bx).
\end{eqnarray*}
The desired conditions of Lemma \ref{lem:two} can be checked as follows.
First observe that with respect to $Q_{\max}(0)$, the function $h$
scales as $b_2(Q_{\max}(0)) f (Q_{\max}(0))$ due to
$b_2/b_1=\Theta\left(\log Q_{\max}(0)\right)$ as per Lemma \ref{lem:adiabetic1}.
Further, $h$ is a function that is lower bounded and its value goes to $\infty$
as $Q_{\max}(0)$ goes to $\infty$. Therefore, $h/g$ scales as $f (Q_{\max}(0))$.
These propeties will imply the verification conditions of Lemma \ref{lem:two}.

\vspace{.1in}
\subsubsection{Step Three: Buffered Circuit Switched Network}\label{sssec:step3switch}

In this section, we prove Lemma \ref{lem:two} for the circuit switched network model.
Similar to wireless network, we are interested in large enough $Q_{\max}(0)$ that
satisfies condition of Lemma \ref{lem:adiabetic1}. Given the state $X(t) = (\bQ(t), \z(t))$
of the Markov process, we shall consider the following Lyapunov function :
$$ L(X(t)) = \sum_i F(R_i(t)).$$
Here $\bR(t) = [R_i(t)]$ with $R_i(t) = Q_i(t)+ z_i(t)$ and as before
$F(x) = \int_{0}^x f(y) ~dy$.  Now we proceed towards finding appropriate
functions $h$ and $g$ as desired in Lemma \ref{lem:two}. For any $\tau \in \Zp$,
\begin{eqnarray}
& &  L(X(\tau+1)) - L(X(\tau)) \nonumber \\
  & & \quad =  \left(F(\bR(\tau+1)) - F(\bR(\tau))\right) \cdot \bone \nonumber\\
    & & \quad \leq f(\bR(\tau+1)) \cdot (\bR(\tau+1)-\bR(\tau)),  \nonumber\\
    & & \quad=   f(\bR(\tau)+A(\tau, \tau+1) - D(\tau, \tau+1)) \cdot \left(A(\tau, \tau+1) - D(\tau, \tau+1)\right) \nonumber\\
    & & \quad \leq   f(\bR(\tau)) \cdot \left(A(\tau, \tau+1) - D(\tau, \tau+1)\right)+\|A(\tau, \tau+1) - D(\tau, \tau+1)\|_2^2. \qquad \qquad \qquad \nonumber
  \end{eqnarray}
Given initial state $X(0) = \bx$, taking expectation for $\tau, \tau+1 \in I$, we have
\begin{eqnarray}\label{eq:ad3}
  && \E_{\bx}[L(X(\tau+1)) - L(X(\tau))]\nonumber\\
    & & \quad \leq  \E_{\bx}\left[f(\bR(\tau)) \cdot A(\tau, \tau+1)\right]- \E_{\bx}\left[f(\bR(\tau)) \cdot D(\tau, \tau+1)\right] \nonumber \\
 & & \qquad \qquad     +\E_{\bx}\left[\|A(\tau, \tau+1)-D(\tau, \tau+1)\|_2^2\right] \nonumber\\
& & \quad {=} \E_{\bx}\left[f(\bR(\tau))\cdot\lamb\right]-
\E_{\bx}\left[f(\bR(\tau)) \cdot D(\tau, \tau+1)\right]+O(1).
\end{eqnarray}
The last equality follows from the fact that arrival process is Poisson with
rate vector $\lamb$ and $\bR(\tau)$ is independent of $A(\tau, \tau+1)$.
In addition, the overall departure process for any $i$, $D_i(\cdot)$, is governed by a
Poisson process of rate at most $C_{\max}$. Therefore, the second moment
of the difference of arrival and departure processes in unit time
is $O(1)$. Now,
\begin{eqnarray}\label{eq:ad3-1}
\E_{\bx}\left[f(\bR(\tau))\cdot\lamb\right] &=& f(\bR(0))\cdot\lamb +\E_{\bx}\left[(f(\bR(\tau))-f(\bR(0)))\cdot\lamb\right].
\end{eqnarray}
And,
\begin{eqnarray}\label{eq:ad3-2}
& & \E_{\bx}\left[f(\bR(\tau)) \cdot D(\tau, \tau+1)\right] \nonumber \\
& & \quad = \E_{\bx}\left[f(\bR(0)) \cdot D(\tau, \tau+1)\right] + \E_{\bx}\left[(f(\bR(\tau))-f(\bR(0))) \cdot D(\tau,\tau+1) \right].% \nonumber \\
%& & \quad \leq \E_{\bx}\left[f(\bR(0)) \cdot D(\tau, \tau+1)\right] + %\E_{\bx}\left[\|f(\bR(\tau))-f(\bR(0))\|_1\right].
\end{eqnarray}
The first term on the right hand side in \eqref{eq:ad3-1} can be bounded as
\begin{eqnarray}
f(\bR(0))\cdot\lamb
&\leq&(1-\varepsilon)\left(\max_{\y\in\X}f(\bR(0))\cdot\y\right)\nonumber\\
&\leq&-\frac{3\beps}4\left(\max_{\y\in\X}f(\bR(0))\cdot\y\right)
+\E_{\pi(0)}\left[f(\bR(0)) \cdot \z\right]+O(1),\label{eq:ad4}
\end{eqnarray}
where the first inequality is due to $\lamb\in(1-\varepsilon)\Conv(\X)$ and
the second inequality follows from Lemma \ref{LEM:GOODPI}
with the fact that $|f_i(\bR(\tau))-f_i(\bQ(\tau))|<f(C_{\max})=O(1)$ for all $i$.
On the other hand, the first term in the right hand side of \eqref{eq:ad3-2} can be bounded
below as
\begin{eqnarray}\label{eq:ad5}
\E_{\bx}\left[f(\bR(0)) \cdot D(\tau, \tau+1)\right]
&=&f(\bR(0)) \cdot\E_{\bx}\left[D(\tau, \tau+1)\right]\nonumber\\
%&=&f(\bR(0)) \cdot\int^{\tau+1}_{\tau}\E_{\bx}\left[D'(s)\right]~ds\nonumber\\
&\geq&
f(\bR(0)) \cdot\int^{\tau+1}_{\tau}\E_{\bx}\left[\z(s)\right]~ds\nonumber\\
&=&\int^{\tau+1}_{\tau}\E_{\mu(s)}\left[f(\bR(0)) \cdot \z\right]~ds.
\end{eqnarray}
In above, we have used the fact that $D_i(\cdot)$ is a Poisson process
with rate given by $z_i(\cdot)$. Further, the second term in the right hand side
of\eqref{eq:ad3} can be bounded as follows.
\begin{eqnarray}\label{eq:ad5-1}
\E_{\bx}\left[\|f(\bR(\tau))-f(\bR(0))\|_1 \right]
&\leq&\E_{\bx}\left[f\left(|\bR(\tau))-\bR(0)|\right) \right] + O(1) \nonumber\\
&\leq&f\left(\E_{\bx}\left[|\bR(\tau))-\bR(0)|\right]\right) +O(1)\nonumber\\
&\leq& n f(C_{\max}\tau) + O(1)\nonumber \\
& = & O(f(\tau)),
\end{eqnarray}
The first inequality follows from $f(x+y) \leq f(x) + f(y) + 2$ for any $x, y \in \Rp$.
This is because $\log (x+y+e) \leq \log (x+e) + \log (y+e)$ for any $x, y \geq \Rp$,
$\log a+b \leq 2 + \log a + \log b$ for any $a, b \geq 1$ and $f(x) = \log \log (x+e)$.
The second inequality follows by applying Jensen's inequality for concave function $f$.
Combining \eqref{eq:ad3}-\eqref{eq:ad5-1}, we obtain
\begin{eqnarray*}
&&\E_{\bx}[L(X(\tau+1)) - L(X(\tau))]\\
&\leq& -\frac{3\beps}4\left(\max_{\y\in\X}f(\bR(0))\cdot\y\right)
+\E_{\pi(0)}\left[f(\bR(0)) \cdot \z\right] \\
& & \qquad -\int^{\tau+1}_{\tau}\E_{\mu(s)}\left[f(\bR(0)) \cdot \z\right]~ds+O(f(\tau))\\
      &\leq&-\frac{3\beps}4\left(\max_{\y\in\X}f(\bR(0))\cdot\y\right) \nonumber \\
& & \qquad      + ~ \int^{\tau+1}_{\tau}\left(\max_{\y\in\X}f(\bR(0))\cdot\y\right)
     \|\mu(s)-\pi(0)\|_{TV}~ds+O(f(\tau))\\
           &\stackrel{(a)}{\leq}&-\frac{\beps}2\left(\max_{\y\in\X}f(\bR(0))\cdot\y\right)+
           O(f(\tau))\\
           &\leq&-\frac{\beps}2f(Q_{\max}(0))+
           O(f(\tau)),
  \end{eqnarray*}
  where (a) follows from Lemma \ref{lem:adiabetic1}. Summing
 this for $\tau\in I=[b_1,b_2-1]$,
\begin{equation}\label{eq:ad8}
  \E_{\bx}[L(X(b_2)) - L(X(b_1))]
    ~\leq~-\frac{\beps}{2}f(Q_{\max}(0))(b_2-b_1)+O((b_2-b_1)f(b_2)).
\end{equation}
Therefore, we have
\begin{eqnarray*}
  &&\E_{\bx}[L(X(b_2)) - L(X(0))]\\
  &=&\E_{\bx}[L(X(b_1)) - L(X(0))]+\E_{\bx}[L(X(b_2)) - L(X(b_1))]\\
    &\stackrel{(a)}{\leq}&\E_{\bx}[f(\bR(b_1)) \cdot (\bR(b_1)-\bR(0))]+\E_{\bx}[L(X(b_2)) - L(X(b_1))]\\
    &=&\sum_i\E_{\bx}[f(R_i(b_1)) \cdot (R_i(b_1)-R_i(0))]+\E_{\bx}[L(X(b_2)) - L(X(b_1))]\\
    &\stackrel{(b)}{\leq}&\sum_i\sqrt{\E_{\bx}[f(R_i(b_1))^2]} \sqrt{\E_{\bx}[(R_i(b_1)-R_i(0))^2]}+\E_{\bx}[L(X(b_2)) - L(X(b_1))]\\
    &\stackrel{(c)}{\leq}&\sum_i\sqrt{f(\E_{\bx}[R_i(b_1)])^2 + O(1)} \cdot O(b_1)+\E_{\bx}[L(X(b_2)) - L(X(b_1))]\\
                &\stackrel{(d)}{=}&n\,f(Q_{\max}(0)+O(b_1))\cdot O(b_1)-\frac{\beps}{2}f(Q_{\max}(0))(b_2-b_1) \\ & & \qquad + ~O((b_2-b_1)f(b_2))\\
                &\stackrel{\triangle}{=}&-h(\bx).
\end{eqnarray*}
Here (a) follows from convexity of $L$; (b) from Cauchy-Schwarz, (c) is due to
the bounded second moment  $\E_{\bx}[|R_i(b_1)-R_i(0)|]=O(b_1)$ as argued earlier
in the proof and observing that there exists a concave function
$g$ such that $f^2 = g + O(1)$ over $\Rp$, subsequently  Jensen's inequality
can be applied; (d) follows from \eqref{eq:ad8}.
Finally, choose $g(\bx) = b_2$.

With these choices of $h$ and $g$, the desired conditions of
Lemma \ref{lem:two} can be checked as follows. First observe
that with respect to $Q_{\max}(0)$, the function $h$
scales as $b_2(Q_{\max}(0)) f (Q_{\max}(0))$ due to
$b_2/b_1=\Theta\left(\log Q_{\max}(0)\right)$ as per Lemma \ref{lem:adiabetic1}.
Further, $h$ is a function that is lower bounded and its value goes to $\infty$
as $Q_{\max}(0)$ goes to $\infty$. Therefore, $h/g$ scales as $f (Q_{\max}(0))$.
These properties will imply the verification conditions of Lemma \ref{lem:two}.

\subsubsection{Step Four}%Proof of Lemma \ref{lem:three}}\label{ssec:lemthree}

For completing the proof of the positive Harris recurrence of both algorithms,
it only remains to show that for $\kappa>0$,
the set $B_\kappa = \{ \bx \in \sX : L(\bx) \leq \kappa \}$ is a closed petit. This
is because other conditions of Lemma \ref{lem:one} follow from Lemma \ref{lem:two}.
And the Step Three exhibited choice of Lyapunov function $L$ and
desired `drift' functions $h, g$.

To this end, first note that $B_{\kappa}$ is closed by definition.
To establish that it is a petit set, we need to find a non-trivial measure $\mu$
on $(\sX, \cB_\sX)$ and sampling distribution $a$ on $\Zp$ so that for any $\bx \in B_\kappa$,
$$ K_a(\bx, \cdot) \geq \mu(\cdot).$$
To construct such a measure $\mu$, we shall use the following Lemma.
\begin{lemma}\label{lem:reachzero}
Let the network Markov chain $X(\cdot)$ start with the state $\bx \in B_\kappa$ at time $0$ i.e. $X(0)=\bx$.
Then, there exists $T_\kappa \geq 1$ and $\gamma_\kappa > 0$ such that
$$ \sum_{\tau=1}^{T_\kappa}  {\Pr}_{\bx}(X(\tau) = \bzero) \geq \gamma_\kappa, ~~\forall \bx \in B_\kappa.$$
Here $\bzero = (\bzero, \bzero) \in \sX$ denote the state where all components of $\bQ$ are
$0$ and the schedule is the empty independent set.
\end{lemma}
\begin{proof}
We establish this for wireless network. The proof for circuit switched network
is identical and we skip the details.  Consider any $\bx \in B_\kappa$. Then
by definition $L(\bx) \leq \kappa + 1$ for given $\kappa > 0$. Hence by definition of $L(\cdot)$ it can be easily checked that each queue is bounded above by $\kappa$. Consider some large enough (soon to be determined)
$T_\kappa$. By the property of Bernoulli
(or Poisson for circuit switched network) arrival process, there is a positive
probability $\theta^0_\kappa > 0$ of no arrivals happening to the system during
time interval of length $T_\kappa$. Assuming that no arrival happens,
we will show that in large enough time $t^1_\kappa$,
with  probability $\theta^1_\kappa > 0$  each queue receives at least $\kappa$ amount of service; and after that in additional time $t^2$ with positive probability $\theta^2 > 0$ the empty set schedule is reached. This will imply that by defining  $T_\kappa \stackrel{\triangle}{=} t^1_\kappa + t^2$
the state $\bzero \in \sX$ is reached  with probability at least
$$\gamma_\kappa  \stackrel{\triangle}{=} \theta^0_\kappa \theta^1_\kappa \theta^2 > 0.$$
And this will immediately imply the desired result of Lemma \ref{lem:reachzero}.
To this end, we need to show existence of $t^1_\kappa, \theta^1_\kappa$ and $t^2, \theta^2$ with
properties stated above to complete the proof of Lemma \ref{lem:reachzero}.

First, existence of $t^1_\kappa, \theta_\kappa^1$. For this, note that the Markov chain corresponding
to the scheduling algorithm has time varying transition probabilities and is irreducible over the space
of all independent sets, $\cI(G)$. If there are no new arrivals and initial $\bx \in B_\kappa$, then clearly
queue-sizes are uniformly bounded by $\kappa$. Therefore, the transition
probabilities of all feasible transitions for this time varying Markov chain is uniformly lower bounded
by a strictly positive constant (dependent on $\kappa, n$). It can be easily checked that the transition
probability induced graph on $\cI(G)$ has diameter at most $2n$ and Markov chain transits as per Exponential clock of overall rate $n$. Therefore, it follows that starting from any initial scheduling configuration,
there exists finite time $\hat{t}_\kappa$ such that a schedule is reached so that any given queue $i$ is scheduled for at least unit amount of time with probability at least
$\hat{\theta}_\kappa > 0$. Here, both $\hat{t}_\kappa, \hat{\theta}_\kappa$ depend on $n, \kappa$. Therefore, it follows that in time $t^1_\kappa \stackrel{\triangle}{=} \kappa n \hat{t}_\kappa$
all queues become empty with probability at least
$\theta^1_\kappa \stackrel{\triangle}{=} \left(\hat{\theta}_\kappa\right)^{n\kappa}$. Next, to
establish existence of $t^2, \theta^2$ as desired, observe that once the
system reaches empty queues, it follows that in the absence of new arrivals the empty schedule $\bzero$
is reached after some finite time $t^2$ with probability $\theta^2 > 0$ by similar properties
of the Markov chain on $\cI(G)$ when all queues are $0$. Here $t^2$ and $\theta^2$ are
dependent on $n$ only. This completes the proof of Lemma \ref{lem:reachzero}.
\end{proof}

In what follows, Lemma \ref{lem:reachzero} will be used to complete the proof that
$B_\kappa$ is a closed petit.
To this end, consider Geometric($1/2$) as the sampling
distribution $a$, i.e.
$$ a(\ell) = 2^{-\ell}, ~~\ell \geq 1.$$
Let $\bdelta_\bzero$ be the Dirac distribution on element $\bzero \in \sX$. Then, define $\mu$ as
$$ \mu = 2^{-T_\kappa} \gamma_k \bdelta_\bzero, ~~\mbox{that is}~~\mu(\cdot) = 2^{-T_\kappa} \gamma_k \bdelta_\bzero(\cdot).$$
Clearly, $\mu$ is non-trivial measure on $(\sX, \cB_\sX)$.
With these definitions of $a$ and $\mu$, Lemma \ref{lem:reachzero} immediately implies
that for any $\bx \in B_\kappa$,
$$ K_a(\bx, \cdot) \geq \mu(\cdot).$$
This establishes that set $B_\kappa$ is a closed petit set.

\section{Discussion}\label{sec:discuss}

% Summarize the results.

This paper introduced a new randomized scheduling algorithm
for two constrained queueing network models: wireless network
and buffered circuit switched network. The algorithm is
simple, distributed, myopic and throughput optimal. The main
reason behind the throughput optimality property of the
algorithm is two folds: (1) The relation of algorithm dynamics
to the Markovian dynamics over the space of schedules that have
a certain product-form stationary distribution, and (2) choice
of slowly increasing weight function $\log\log(\cdot+e)$ that
allows for an effective time scale separation between algorithm
dynamics and the queueing dynamics.  We chose wireless network
and buffered circuit switched network model to explain the
effectiveness of our algorithm because (a) they are becoming
of great interest \cite{mesh,optical} and (b) they represent
two different, general class of network models: synchronized
packet network model and asynchronous flow network model.

% Discuss distributed implementation
%% there is a simple algorithm : describe quality/mp/simple/thputopt + reference.

Now we turn to discuss the distributed implementation of
our algorithm. As described in Section \ref{ssec:algo1}, given
the weight information at each wireless node (or ingress of
a route), the algorithm completely distributed. The weight,
as defined in \eqref{eq:weight1} (or \eqref{eq:weight2}),
depends on the local queue-size as well as the $Q_{\max}$
information. As is, $Q_{\max}$ is global information. To
keep the exposition simpler, we have used the precise
$Q_{\max}$ information to establish the throughput
property. However, as remarked earlier in the Section
\ref{ssec:algo1} (soon after \eqref{eq:weight1}), the $Q_{\max}$ can
be replaced by its appropriate distributed estimation
without altering the throughput optimality property. Such
a distributed estimation can be obtained through an extremely
simple Markovian like algorithm that require each node
to perform broadcast of exactly one number in unit
time. A detailed description of such an algorithm can be
found in Section 3.3 of \cite{RSS09}.

On the other hand, consider the algorithm that does not
use $Q_{\max}$ information. That is, instead of
\eqref{eq:weight1} or \eqref{eq:weight2}, let weight be
$$ W_i(t) = f(Q_i(\lfloor t\rfloor)).$$
We conjecture that this algorithm is throughput
optimal.

\section*{Acknowledgements}
We would like to acknowledge the support of the NSF projects CNS 0546590,
TF 0728554 and DARPA ITMANET project.

\bibliographystyle{plain}

\bibliography{biblio}
\begin{appendix}

\section{A Useful Lemma}\label{a0}

\begin{lemma}\label{lem:last}
Let $P_1, P_2 \in \R^{N\times N}$. Then,
$$\left\|e^{P_1}-e^{P_2}\right\|_{\infty}\leq
e^{NM}\left\|P_1-P_2\right\|_{\infty},$$
where $M=\max\{\|P_1\|_{\infty},\|P_1\|_{\infty}\}$.
\end{lemma}
\begin{proof}
Using mathematical induction, we first establish that for any $k\in \N$,
\begin{eqnarray}\label{eq:lasteq1}
\|P_1^k- P_2^k\|_{\infty} & \leq & k (N M)^{k-1}\|P_1-P_2\|_{\infty}.
\end{eqnarray}
To this end, the base case $k=1$ follows trivially. Suppose it is
true for some $k \geq 1$. Then, the inductive
step can be justified as follows.
\begin{eqnarray*}
\|P_1^{k+1}- P_2^{k+1}\|_{\infty}
&=&\left\|P_1\left(P_1^{k}- P_2^{k}\right)
+\left(P_1-P_2\right)P_2^k\right\|_{\infty}\\
&\leq&\left\|P_1\left(P_1^{k}- P_2^{k}\right)\right\|_{\infty}
+\left\|\left(P_1-P_2\right)P_2^k\right\|_{\infty}\\
&\stackrel{(a)}{\leq}&N\,\left\|P_1\right\|_{\infty}\left\|P_1^{k}- P_2^{k}\right\|_{\infty}
+N\,\left\|P_1-P_2\right\|_{\infty}\left\|P_2^k\right\|_{\infty}\\
&\stackrel{(b)}{\leq}&N M\times k (N\,M)^{k-1}\|P_1-P_2\|_{\infty}
+N\,\left\|P_1-P_2\right\|_{\infty}\times N^{k-1}M^k\\
&=&(k+1) (N M)^{k}\|P_1-P_2\|_{\infty}.
\end{eqnarray*}
In above, (a) follows from an easily verifiable fact that for any $Q_1, Q_2 \in \R^{N\times N}$,
$$\|Q_1Q_2\|_{\infty}\leq N\|Q_1\|_{\infty}\|Q_2\|_{\infty}.$$
We use induction hypothesis to justify (b). Using \eqref{eq:lasteq1}, we
have
\begin{eqnarray*}
\left\|e^{P_1}-e^{P_2}\right\|_{\infty}
&=& \left\|\sum_k \frac1{k!} \left(P_1^k-P_2^k\right)\right\|_{\infty}\\
&\leq& \sum_k \frac1{k!}\left\|P_1^k-P_2^k\right\|_{\infty}\\
&\leq& \sum_k \frac1{k!}k (N M)^{k-1}\|P_1-P_2\|_{\infty}\\
&=& e^{NM}\|P_1-P_2\|_{\infty}.
\end{eqnarray*}

\end{proof}
\end{appendix}
\end{document}